\documentclass[11pt,english]{smfart}
\usepackage{graphicx}
\usepackage{hyperref}
\usepackage{epstopdf}
\usepackage{a4wide}
\usepackage{amssymb}
\usepackage{amsfonts}
\usepackage{amsmath}
\usepackage{xcolor}
\usepackage[T1]{fontenc}
\usepackage[latin1]{inputenc}
\usepackage{caption}
\usepackage{subcaption}

\newtheorem{theorem}{Theorem}[section]
\newtheorem{lemma}[theorem]{Lemma}

\newtheorem{definition}[theorem]{Definition}
\newtheorem{remark}[theorem]{Remark}
\def\di{\displaystyle}
\newcommand{\R}{\mathbb{R}}

\newcommand{\tR}{\tilde{R}}
\newcommand{\tT}{\tilde{T}}

\newcommand{\bR}{\bar{R}}
\newcommand{\bT}{\bar{T}}

\newcommand{\tAx}{\tilde{A}_x}
\newcommand{\tAy}{\tilde{A}_y}

\begin{document}

%%-----------------------------------------------------------------
%%      the top matter
%%

\setcounter{tocdepth}{3}

\title{The Sharma-Parthasarathy stochastic two-body problem}
\author{J. Cresson$^{1,2}$, F. Pierret$^2$, B. Puig$^3$}
\address{$1$ LMAP/Universit\'e de Pau, 64013 Pau, France, $2$ SYRTE/Observatoire de Paris, 75014 Paris, France, $3$ IPRA/Universit\'e de Pau, 64013 Pau.}

\begin{abstract}
We study the Sharma-Parthasarathy stochastic two-body problem introduced by N. Sharma and H. Parthasarathy in \cite{sharma}. In particular, we focus on the preservation of some fundamental features of the classical two-body problem like the Hamiltonian structure and first integrals in the stochastic case. Numerical simulations are performed which illustrate the dynamical behaviour of the osculating elements as the semi-major axis, the eccentricity and the pericenter. We also derive a stochastic version of Gauss's equations in the planar case.
\end{abstract}

\maketitle

{\bf Keywords}: Two-body problem, stochastic perturbation, numerical simulations, stochastic Gauss's equations.

%%-----------------------------------------------------------------

%\tableofcontents

\section{Introduction}
%%---------------------

The aim of this paper is to study a stochastic perturbation of the two-body problem introduced by S.N. Sharma and H. Parthasarathy in \cite{sharma} both theoretically and numerically. The perturbation constructed by Sharma and al. is designed to model the force induced by a cloud having a density which fluctuates stochastically. This assumption is supported by observations made by \cite{manndust} about the zodiacal dust around the sun. It must be noted that other examples of stochastic perturbations of the two-body problem can be found in the literature as for example in  (\cite{albeverio1983,albeverio1984,nottale1995,nottale1996,Cr,zambrini,CD1,nelson1966,nelson2001,Mumford99thedawning}) which is not an exhaustive list. However, they do not consider the situation covered by this model.\\

The paper of S.N. Sharma and H. Parthasarathy \cite{sharma} is mainly concerned with constructing a tractable simplified model of the stochastic equations which accurately reproduces the behaviour of the orbiting particle. The classical linearisation procedure around the mean behaviour is in this case ineffective. As a consequence, they develop a second order approximation of the nonlinearity and study the properties of such an approximation.\\ 

In this paper, we return to the initial stochastic model in order to understand what are the main differences with respect to the classical features of the two-body problem as for examples the Hamiltonian structure and conserved quantities. Our results are supported by numerical simulations which are obtained using a specific stochastic Runge-Kutta introduced by . Finally, we derive stochastic equations for the behaviour of the orbital elements. As pointed out in \cite{sharma}, these quantities are fundamental for an accurate positioning of the orbiting particle. We give also some numerical simulation illustrating the resulting behaviour of the orbital elements.\\ 

The plan of the paper is as follows : In Section \ref{sharma-model} we define the Sharma-Parthasarathy stochastic two-body problem following \cite{sharma}. Section \ref{hamiltonian} discuss the preservation of the Hamiltonian structure in the stochastic case and Section \ref{integrals} deals with the behaviour of first integrals. Using a specific numerical method, we perform simulations in Section \ref{simul}. In Section \ref{stochastic-gauss}, we derive the stochastic Gauss equations, i.e. the equations governing the behaviour of the orbital elements and we perform numerical simulations. Finally, Section \ref{conclusion} gives our conclusion and perspectives. 
  
\section{The Sharma-Parthasarathy stochastic two-body problem}
\label{sharma-model}  
\subsection{Reminder about stochastic differential equations}

We remind basic properties and definition of stochastic differential equations in the sense of It\^o. We refer to the book \cite{oksendal2003stochastic} for more details.\\

A {\it stochastic differential equation} is formally written (see \cite{oksendal2003stochastic},Chap.V) in differential form as  
\begin{equation}
\label{stocequa}
dX_t = \mu (t,X_t)dt+\sigma(t,X_t)dB_t ,
\end{equation}
which corresponds to the stochastic integral equation
\begin{equation}
\label{stocintegral}
X_t=X_0+\int_0^t \mu (s,X_s)\,ds+\int_0^t \sigma (s,X_s)\,dB_s ,
\end{equation}
where the second integral is an It\^o integral (see \cite{oksendal2003stochastic},Chap.III) and $B_t$ is the classical Brownian motion (see \cite{oksendal2003stochastic},Chap.II,p.7-8).\\

An important tool to study solutions to stochastic differential equations is the {\it multi-dimensional It\^o formula} (see \cite{oksendal2003stochastic},Chap.III,Theorem 4.6) which is stated as follows : \\

We denote a vector of It\^o processes by $\mathbf{X}_t^\mathsf{T} = (X_{t,1}, X_{t,2}, \ldots, X_{t,n})$ and we put $\mathbf{B}_t^\mathsf{T} = (B_{t,1}, B_{t,2}, \ldots, B_{t,n})$to be a $n$-dimensional Brownian motion (see \cite{karatzas},Definition 5.1,p.72),  $d\mathbf{B}_t^\mathsf{T} = (dB_{t,1}, dB_{t,2}, \ldots, dB_{t,n})$. We consider the multi-dimensional stochastic differential equation defined by (\ref{stocequa}). Let $f$ be a $\mathcal{C}^2(\mathbb{R}_+ \times \mathbb{R},\mathbb{R})$-function and $X_t$ a solution of the stochastic differential equation (\ref{stocequa}). We have 
\begin{eqnarray}
df(t,\mathbf{X}_t) = \frac{\partial f}{\partial t} dt + (\nabla_\mathbf{X}^{\mathsf T} f) d\mathbf{X}_t + \frac{1}{2} (d\mathbf{X}_t^\mathsf{T}) (\nabla_\mathbf{X}^2 f) d\mathbf{X}_t,
\end{eqnarray}
where $\nabla_\mathbf{X} f = \partial f/\partial \mathbf{X}$ is the gradient of $f$ w.r.t. $X$, $\nabla_\mathbf{X}^2 f = \nabla_\mathbf{X}\nabla_\mathbf{X}^\mathsf{T} f$ is the Hessian matrix of $f$ w.r.t. $\mathbf{X}$, $\delta$ is the Kronecker symbol and the following rules of computation are used : $dt dt = 0$, $dt dB_{t,i}  = 0$, $dB_{t,i} dB_{t,j} = \delta_{ij} dt$.
  
\subsection{The Sharma-Parthasarathy stochastic two-body problem}
%%-------------------------
\label{two-body-intro}

In \cite{sharma} the authors consider a stochastic perturbation of the two-body problem induced by a cloud with a density which fluctuates stochastically. This assumption is supported by observations made by \cite{manndust} about the zodiacal dust around the Sun.\\ 

Let $S$ and $P$ be two bodies and $M_S$ and $M_P$ their masses. The body $S$ is supposed to be the central body typically a star and $P$ is the orbiting body typically a planet or a satellite. The motion is supposed to be in an elliptic configuration. The reduced mass is $m=\frac{M_S M_P}{M_S + M_P}$ and the potential coefficient is $k=G M_S M_P$ where $G$ is the gravitational constant. We define $(S,\vec{x},\vec{y})$ to be a fixed frame attached to $S$ and $\vec{r}$ the position vector of $P$ in this reference frame with $\phi$ his position angle. The elliptical motion is described with the semi-major axis $a$, the eccentricity $e$ and the pericenter angle $\omega$. We associate the polar reference frame $(S,\vec{e_R},\vec{e_T})$ where  $\vec{e_R}^\mathsf{T} = (\cos \phi, \sin \phi )$ and $\vec{e_T}^\mathsf{T} = (-\sin \phi, \cos \phi )$. In this reference frame we have $\vec{r}=r \vec{e_R}$ where $r$ is the norm of the position vector. The motion is illustrated in Fig.~\ref{pierret:fig2body}

\begin{figure}[ht!]
 \centering
 \includegraphics[width=0.32\textwidth,clip]{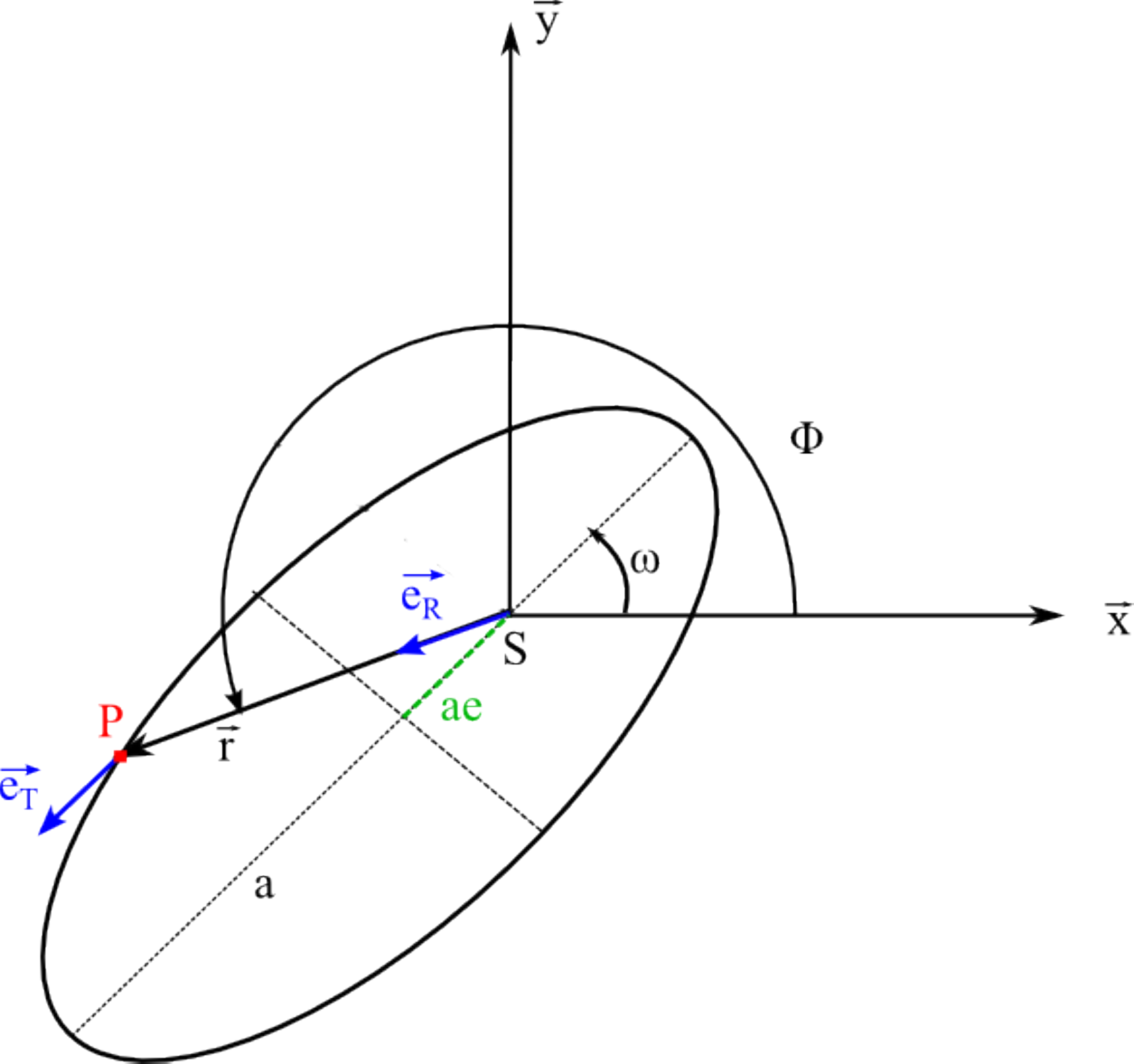}      
%% Note the ABSENCE of the extension .pdf , .eps or .ps  !
  \caption{The classical two body problem.}
  \label{pierret:fig2body}
\end{figure}

The general form of the equations of the perturbed two-body problem by a planar force $\vec{F} = (F_r, F_\phi)$ is easily computed (see \cite{goldstein},Chap.3) and reads
\begin{equation}
\left\{
\begin{array}{r c l}
\frac{dr}{dt} &=& v ,\\
\frac{d\phi}{dt} &=& w, \\
\frac{dv}{dt} &=& rw^2-\frac{k}{m r^2} + \frac{F_r}{m} ,\\
\frac{dw}{dt} &=& -\frac{2 v w}{r} + \frac{F_\phi}{m r} .
\end{array}
\right.
\end{equation}
In \cite{sharma}, the authors take \begin{eqnarray}
\label{randomforce}
\vec{F}^\mathsf{T} = \left( m r \sigma_r W^r_t, m \sigma_\phi W^\phi_t \right)
\end{eqnarray}
where $\sigma_\phi$ is a constant and $W^\phi_t$ is also a "white noise" independent of $W^r_t$ leading to the following stochastic differential system where the {\it white noise process} is heuristically obtained as the "derivative" of the {\it Brownian motion} $B_t$ (see \cite{oksendal2003stochastic},p.7-8) :
\begin{equation}
\label{probleme-2corps-stochastique}
\left\{
\begin{array}{lcl}
d r &=& v dt ,\\
d \phi &=& w dt ,\\
d v &=& \left(rw^2-\frac{k}{m r^2}\right)dt+ r \sigma_r dB^r_t ,\\
d w &=& -\frac{2 v w}{r}dt + \frac{\sigma_\phi}{r} dB^{\phi}_t ,
\end{array}
\right.
\end{equation}
where $B^r_t$ and $B^{\phi}_t$ are independent. This set of equations describes what we called the {\it Sharma-Parthasarathy stochastic two-body problem} in the following.

\section{Hamiltonian structure and first integrals}

In this Section, we discuss the preservation of some fundamental features of the two body problem for the Sharma-Parthasarathy stochastic two-body problem. These information are useful to validate our simulations of the system in the next Section.

\subsection{Hamiltonian structure}
\label{hamiltonian}

Before studying the preservation of the Hamiltonian structure, we remind a stochastic analogue introduced by J-M. Bismut in \cite{bismut1981meca} and called {\it stochastic Hamiltonian systems}.

\subsubsection{Stratonovich stochastic differential equations}

Stochastic Hamiltonian systems are defined in the framework of the Stratonovich interpretation of stochastic differential equations. We refer to \cite{oksendal2003stochastic} for more details.\\

A Stratonovich stochastic differential equation is formally denoted in differential form by
\begin{equation}
\label{strato}
dX_t = \mu (t,X_t ) dt +\sigma (t,X_t ) \circ dB_t ,
\end{equation}
which corresponds to the stochastic integral equation 
\begin{equation}
\label{equa}
X_t =x+\di\int_0^t \mu (s, X_s ) ds + \di\int_0^t \sigma (s, X_s ) \circ dB_t , 
\end{equation}	
where the second integral is a Stratonovich integral (see \cite{oksendal2003stochastic},p.24,2)).\\

Solutions of the Stratonovich differential equation (\ref{strato}) corresponds to the solutions of a modified It\^o equation (see \cite{oksendal2003stochastic},p.36) :
\begin{equation}
\label{ito}
dX_t = \mu_{\rm cor} (t,X_t ) dt +\sigma (t,X_t )  dB_t ,
\end{equation}
where
\begin{equation}
\mu_{\rm cor} (t,x ) =\left [ \mu (t,x ) +\di\frac{1}{2} \sigma' (t,x ) \sigma (t,x) \right ] .
\end{equation}
The correction term $\di\frac{1}{2} \sigma' (t,X_t ) \sigma (t,X_t)$ is also called the {\it Wong-Zakai} correction term (see \cite{stroock}). 

In the multidimensional case, i.e. $\mu :\R^{n+1} \rightarrow \R^n$, $\mu (t,x)=(\mu_1 (t,x) ,\dots ,\mu_n (t,x))$ and $\sigma :\R^{n+1} \rightarrow \R^{n\times p}$, 
$\sigma (t,x)=( \sigma_{i,j} (t,x))_{1\leq i\leq n,\ 1\leq j\leq p}$ the analogue of this formula is given by (see \cite{oksendal2003stochastic},p.85) :
\begin{equation}
\label{strato-ito}
\mu_{\rm cor, i} (t,x)=\mu_i (t,x)+\di\frac{1}{2} \di\sum_{j=1}^p \di\sum_{k=1}^n \di\frac{\partial \sigma_{i,j}}{\partial x_k} \sigma_{k,j},\ \ 1\leq i \leq n .
\end{equation}

The main advantage of the Stratonovich integral is that it induces classical chain rule formulas under a change of variables.

\subsubsection{Reminder about stochastic Hamiltonian systems}
\label{section-hamiltonian}
In the following we deal with stochastic differential equations in the Stratonovich sense. \\

Stochastic Hamiltonian systems are formally defined as :

\begin{definition}
A stochastic differential equation is called stochastic Hamiltonian system if we can find a finite family of functions $\mathbf{H}=\left \{ H_r \right \}_{r=0,\dots ,m}$, $H_r :\R^{2n} \mapsto \R$, $r=0,\dots ,m$ 
such that
\begin{equation}
\left\{
\begin{array}{lll}
dP^i & = & -\di\frac{\partial H}{\partial q_i} dt - \displaystyle\sum_{r=1}^m \frac{\partial H_r}{\partial q_i} (t,P,Q) {\circ} d B^r_t ,\\
dQ^i & = & \di\frac{\partial H}{\partial p_i} dt + \displaystyle\sum_{r=1}^m \frac{\partial H_r}{\partial p_i} (t,P,Q) {\circ} d B^r_t .
\end{array}
\right.
\end{equation}
\end{definition}

We recover the classical algebraic structure of Hamiltonian systems. The main properties supporting this definition are the following one, already proved in Bismut \cite{bismut1981meca} :\\

\begin{itemize}
\item {\it Liouville's property} : Let $(P,Q)\in \R^{2n}$, we consider the stochastic differential equation
\begin{equation}
\label{eq**}
\left\{
\begin{array}{lll}
dP & = & f(t,P,Q)dt + \di\sum_{r=1}^m \sigma_r (t,P,Q) \circ d B^r_t ,\\
dQ & = & g(t,P,Q)dt + \di\sum_{r=1}^m \gamma_r (t,P,Q) \circ d B^r_t .
\end{array}
\right.
\end{equation}
The phase flow of (\ref{eq**}) preserves the {\it symplectic structure} if and only if it is a stochastic Hamiltonian system.	\\

\item {\it Hamilton's principle} : Solutions of a stochastic Hamiltonian system correspond to critical points of a stochastic functional defined by 
\begin{equation}
\mathcal{L}_{\mathbf{H}} (X)=\di \int_0^t H_0(s,X_s ) + \sum_{r=1}^{m} H_r(s,X_s )\circ d B^r_t .
\end{equation}
\end{itemize}

\subsubsection{Is the stochastic two-body problem Hamiltonian ?}

In order to determine if the stochastic two-body problem possess or not a stochastic Hamiltonian structure, we derive the Stratonovich form of our equations.

\begin{theorem}
The Stratonovich form of the stochastic two-body problem is given by  
\begin{equation}
\left\{
\begin{array}{lll}
dr & = &  \frac{p_r}{m} dt , \\
d\phi & = &  \frac{p_\phi}{m r^2}dt ,\\
d p_r & = &  (\frac{p_\phi^2}{m r^3} - \frac{k}{r^2} )dt +  m \sigma_r r \circ d B^r_t , \\
dp_{\phi} & = &  m \sigma_\phi r \circ d B_t^{\phi} .
\label{eq8}
\end{array}
\right.
\end{equation}
\label{thm2}
\end{theorem}
		
\begin{proof}
Using formula (\ref{strato-ito}), we easily prove that the Wong-Zakai correction term is zero. 
\end{proof}
		
As a consequence, the Itô and Stratonovich formulations coincide for this model.\\

We are now in position to answer our question about the persistence of the Hamiltonian structure under the stochastic perturbation. We use the following characterization of stochastic Hamiltonian systems due to Milstein and al. \cite{milstein} :

\begin{theorem}
A $2n$-system of stochastic differential equations of the form 
\begin{equation}
\left\{
\begin{array}{lll}
dP & = &  f(t,P,Q)dt + \di\sum_{r=1}^m \sigma_r (t,P,Q) \circ d B^r_t , \\
dQ & = &  g(t,P,Q)dt + \di\sum_{r=1}^m \gamma_r (t,P,Q) \circ d B^r_t ,
\end{array}
\right.
\label{eq*}
\end{equation}
possesses a stochastic Hamiltonian formulation if and only if the coefficients satisfy the following set of conditions
\begin{equation}
\label{condition-hamiltonian}
\begin{array}{lll}
\di\frac{\partial \sigma_{ir}}{\partial p^\alpha} + \di\frac{\partial \gamma_{\alpha r}}{\partial q^i} & = &  0 , \\
\di\frac{\partial \sigma_{ir}}{\partial q^\alpha} & = &  \di\frac{\partial \sigma_{\alpha r}}{\partial q^i},\ \alpha \neq i ,\\
\di\frac{\partial \gamma_{ir}}{\partial p^\alpha} & = &  \di\frac{\partial \gamma_{r\alpha}}{\partial p^i},\ \alpha \neq i ,
\end{array}
\end{equation}
for $i,\alpha=1,...n$.
\label{thmilsteincond}
\end{theorem}

Simple computations lead to :
		
\begin{theorem}
The stochastic two-body problem does not possess a stochastic Hamiltonian formulation.
\label{thm3}
\end{theorem}
		
\begin{proof}
We use Theorem \ref{thmilsteincond} for the system (\ref{probleme-2corps-stochastique}). As $\gamma$ is null and $\sigma = \left( \begin{matrix} m \sigma_r r & 0 \\ 0 & m \sigma_\phi r \end{matrix} \right) $ does not depends on the conjugate variables $v$ and $w$, the first two conditions of (\ref{condition-hamiltonian}) are trivially satisfied. 

The last condition is equivalent to $\frac{\partial \sigma_{11}}{\partial \phi} = \frac{\partial \sigma_{21}}{\partial r}$ and $\frac{\partial \sigma_{22}}{\partial r} = \frac{\partial \sigma_{12}}{\partial \phi}$. The first equation is satisfied and the second one reduces to 
\begin{equation}
m \sigma_\phi = 0,
\end{equation}
which is satisfied if and only if $\sigma_\phi =0$, i.e. there is no tangential component to the noise, which is not allowed in our model. 
\end{proof}

\subsection{Symmetries and First integrals}
\label{integrals}
First integrals and symmetries play a fundamental role in classical mechanics and in particular for the study of the deterministic $n$-body problem (see \cite{arnold}). A natural question is to know if symmetries and first integrals of a given deterministic system persist in an appropriate sense. In this Section, we remind the definition of weak and strong first integrals as introduced for example by M. Thieullen and J.C. Zambrini (\cite{tz},\cite{thieullen-zambrini1997} or \cite{misawa},\cite{cami},\cite{CD1},\cite{bismut1981meca}). We prove that the angular momentum is preserved under stochastic perturbation and give rise to a weak first integral of the stochastic two-body problem. 

\subsubsection{Definitions}

Let $dx/dt = f(x,t)$, $x\in {\mathbb R}^n$ ($\star$) be an ordinary differential equation. A function $I :{\mathbb R}^n \mapsto {\mathbb R}$ is called a {\it first integral} of ($\star$) if for all solutions $x_t$ of ($\star$) we have $I (x_t ) = I(x_0 )$ for all $t$.  If $I$ is sufficiently smooth we deduce $\frac{dI(x_t)}{dt}=0$.\\

A natural generalisation of this definition in the setting of stochastic differential equations is given for example in \cite{misawa} (see also \cite{tz,thieullen-zambrini1997, bismut1981meca},\cite{CD1},\cite{CD2} and \cite{cami},p.52):

\begin{definition}[Strong first integral]
A function $I :{\mathbb R}^n \rightarrow {\mathbb R}$ is a {\it strong first integral} of (\ref{stocequa}) if for all solutions $X_t$ of (\ref{stocequa}), the stochastic process $I(X_t )$ is a constant process, i.e. $I(X_t ) = I(X_0 )$ a.s. (almost surely) or $d (I(X_t )) =0$. 
\end{definition}

Such a property is very strong and classical first integral are usually not preserved in the strong sense. However, a weaker property can be looked for:

\begin{definition}[Weak stochastic first integral]
A function $I :{\mathbb R}^n \rightarrow {\mathbb R}$ is a weak stochastic first integral of (\ref{stocequa}) if for all solutions $X_t$ of (\ref{stocequa}), the stochastic process $I(X_t )$ satisfies $\mbox{\rm E}(I(X_t))=\mbox{\rm E}(I(X_0 ))$ where $\mbox{\rm E}$ denotes the expectation.
\end{definition}

Of course strong first integrals are also weak first integrals as the equality $I(X_t ) =I(X_0 )$ a.s. implies that $\mbox{\rm E}(I(X_t))=\mbox{\rm E}(I(X_0 ))$. 

\subsubsection{Variation of the angular momentum and the energy}

Classical conserved quantities of motion for the two-body problem are the {\it angular momentum} and {\it energy} of the system defined by  
\begin{eqnarray}
M=mr^2 w ,\label{angular} \\
H= \frac{1}{2}m(v^2+r^2w^2) - \frac{k}{r} .\label{energy}
\end{eqnarray}

Using formulas (\ref{angular}) and (\ref{energy}) for the angular momentum and energy, the multi-dimensional It\^o formula with $X_t^\mathsf{T}=\left(r,\phi,v,w \right)$ and $B_t^\mathsf{T}=\left(B^r_t ,B^{\phi}_t \right)$ leads to 	
$$dM(X_t) = m r \sigma_\phi dB^{\phi}_t ,\ \ 
	dH(X_t) = m r v \sigma_r dB^r_t + m r w \sigma_\phi dB^{\phi}_t + \frac{m}{2}\left[\sigma_r^2 r^2 + \sigma_\phi^2\right]dt .
$$
for the behaviour of these first integrals over solutions of the stochastic two-body problem. As expected, there is no persistence of the angular momentum or energy integral in the strong sense. 

\begin{remark}
The strong conservation of the angular momentum is broken by our assumption that an isotropic tangential force exists, i.e. $\sigma_{\phi} \not= 0$. 
\end{remark}

However, we have the following weak conservation property :

\begin{lemma}
The angular momentum is a weak first integral of the stochastic two-body problem.
\end{lemma}

The proof is simple and relies on classical properties of the Brownian motion. 

\begin{proof}
Let $X_t$ be a solution of the stochastic two-body problem. We have $M(X_t) = M(X_0) + \int_{0}^{t} m r\sigma_\phi dB^{\phi}_t$ where $M$ is the angular momentum function. Using the property that $\mbox{\rm E}\left(\int_a^b f dB\right)=0$ for all $f$ sufficiently smooth (see \cite{oksendal2003stochastic},Definition 3.4,p.18 and Theorem 3.7 (iii),p.22), we deduce that $\mbox{\rm E}(M(X_t)) = \mbox{\rm E}(M(X_0))$ which concludes the proof. 
\end{proof}

This result does not extend to the energy first integral. This is due to the existence of a non-trivial deterministic term emerging in the It\^o formula. Precisely, we have $H(X_t) = H(X_0) + \int_{0}^{t} m rv\sigma_r dB^r_t + \int_{0}^{t} mrw\sigma_\phi dB^{\phi}_t + \frac{m}{2} \int_{0}^{t} \left[\sigma_r^2 r^2 + \sigma_\phi^2 \right]ds$. Taking expectation, we obtain 
$$\mbox{\rm E}(H(X_t)) = \mbox{\rm E}(H(X_0 )) + \frac{m}{2} \mbox{\rm E}\left(\int_{0}^{t} \left[\sigma_r^2 r^2 + \sigma_\phi^2 \right]ds \right) .$$
The second term is non zero so that the energy first integral is not preserved even in a weak sense.

\begin{remark}
The conservation of the angular momentum in the weak sense will be an important information in order to perform simulations because it will be the only quantity that we could check his conservation during the simulations.
\end{remark}

\section{Simulations}
\label{simul}
The simulation of stochastic differential equations is more difficult than in the deterministic case (see \cite{kloeden} and \cite{higham}). In the sequel, we use a stochastic Runge-Kutta of weak order 2 due to N.J. Kasdin and L.J. Stankievech in \cite{kasdin}. The term of {\it weak order} refers to the error of the stochastic numerical scheme with respect to the expectation of the solution computed.

\subsection{A stochastic Runge-Kutta method of weak order 2}

\subsubsection{Kasdin and al. stochastic Runge-Kutta method}

The numerical scheme of N.J. Kasdin and L.J. Stankievech in \cite{kasdin} is based on the strategy of construction of Runge-Kutta type methods in the deterministic case which used the Taylor expansion of function in order to determine the coefficient of the scheme. Using the {\it Ito-Taylor expansion} (see \cite{kloeden},Theorem 5.5.1, p.181-182), one can construct in the same way such methods for stochastic differential equations. 

The main difference to construct a Runge-Kutta method of $n$ in the stochastic case versus the deterministic case is the lake of constraining equations to determine coefficients of the method due to the existence of multiple paths. 

The strategy used to bypass this difficulty is to weaken the method in sense of stochastic calculus that is to say to consider only quantities in term of expectation. It reduces considerably the under determined system for the coefficients and improves the development of algorithms with high order in the weak sense. \\

The weak second order method of N.J. Kasdin and L.J. Stankievech in \cite{kasdin} is described as follows :

\begin{align}
x_{n+1} &= x_n + \sum_{l=1}^{2} \alpha_l k_l + \beta_l j_l ,\\
k_1&=h f(x_n,t_n) ,\\
j_1&=g(x_n,t_n)w_1 ,\\
k_2&=h f(x_n+a_{21}k_1+b_{21}j_1,t_n + c_2h) ,\\
j_2&=g(x_n+e_{21}k_1+g_{21}j_1,t_n + d_2h)w_2 ,
\end{align}
where $x=(x_0,...,x_N)$ is the numerical solution with time step $h$ of the time interval $[0,T]$ defined for $i=0,...,N-1$ by $t_{i+1}-t_i=h$ and $w_1,w_2$ are independently and identically distributed Gaussian random numbers such that 

\begin{align}
E(w_l)&=0 ,\\
E(w_l w_m)&=q_l Q h \delta_{lm} ,
\end{align}
where the $q_l$ are additional coefficients defining the variance of each noise sample ,$Q$ is a constant defining the variance of the increments of Brownian motion and $\delta$ is the Kronecker delta function.

N.J. Kasdin and L.J. Stankievech in \cite{kasdin} have two set of coefficients for the method one obtain numerically and the other as a Heun Analog of the deterministic Heun method which allow to reduce the stochastic method to the deterministic method if there is no stochastic perturbation :

\begin{center}
\begin{tabular}{|c|c|c|}
\hline Coefficients & Heun Analog & Numerical Search \\ 
\hline $\alpha_1$ & 1/4 & 0.136713 \\
\hline $\alpha_2$ & 3/4 & 0.863287 \\ 
\hline $\beta_1$ & 1 & -1.512997 \\ 
\hline $\beta_2$ & 1 & 1.112094 \\ 
\hline $c_2$ & 2/3 & 0.579182 \\ 
\hline $d_2 $ & 3/2 & 1.18816  \\ 
\hline $a_{21}$ & 2/3 & 0.579182  \\ 
\hline $b_{21}$ & 1 & -1.512997 \\ 
\hline $e_{21}$ & 3/2 & 1.18816 \\ 
\hline $g_{21}$ & 3/2 & 2.16704 \\ 
\hline $q_1$ & 2/3 & 0.25301 \\ 
\hline $q_2$ & 1/3 & 0.34026 \\ 
\hline 
\end{tabular} 
\end{center}

\subsubsection{Implementation of the method}

The method is implemented in a Fortran program with the coefficients determined by the numerical search because as N.J. Kasdin and L.J. Stankievech in \cite{kasdin} pointed out it works better than the Heun Analog coefficients. In order to compute expectation quantities, several millions of Brownian realization are needed to more accuracy with the Monte Carlo methods. Our program uses parallel distribution of Brownian realization on cluster which considerably reduce the amount of time needed to compute expectation with high accuracy.\\ 

The simulations are performed on a cluster with 56 processor unit.
The cluster is composed of an Intel Xeon CPU E5649, an Intel Xeon X5570, an
Intel Core CPU i7-3720QM and an Intel Core CPU i7-2600K.
It allows us to perform simulations for the test of the weak convergence in
few seconds with 50 000 Brownian motion realization.
For the numerical simulations of the two-body problem we can perform $5.10^6$
Brownian motion realization with a time step $dt=0.01$ and a final time $T=15$ in
about $5$ minutes only.\\

A version of this program in Fortran and Scilab can be downloaded on the web-page of Fr\'ed\'eric {\sc Pierret}. See \begin{verbatim}
http://syrte.obspm.fr/~pierret/two_body_sto.html
\end{verbatim}

\subsubsection{About the weak convergence of the method}

In this Section, we provide numerical results indicating that the Kasdin and al. RK method converge in a weak sense. Our test is done on the classical Orstein-Ulenbeck model for which explicit solutions are known. \\

First we remind the definition of the {\it weak convergence} (see \cite{kloeden},p.326-327 or D.J. Higham \cite{higham},p.537) :

\begin{definition}
A method is said to have a weak order of convergence equal to $\gamma$ if there exists a constant $C$ such that
\begin{align}
|{\rm E}(x_n)-{\rm E}(x(\tau))| \le C h^\gamma ,
\end{align}
where $h$ is the time step, $\tau=n h \in [0,T]$ a fixed point and $x_n$ the numerical solution at time $t_n=nh$.
\end{definition}

Now we consider the {\it Langevin equation}
\begin{equation}
dX_t = - \mu X_t dt + \sigma dB_t ,
\end{equation}
with $\mu,\sigma \in \R$.\\

The expectation and variance are explicitly known for the solution of this equation :

\begin{align}
{\rm E}(X_t) &= X_0 e^{-\mu t} ,\\
Var(X_t) & \equiv {\rm E}(X_t^2) = X_0^2 e^{-2\mu t} + \frac{\sigma^2}{2\mu}(1-e^{-2\mu t}).
\end{align}

The simulations for different values of $\sigma$ show that the numerical method converge in the weak sense. See Figures \ref{analytic-test-0001} and \ref{analytic-test-00001}.

\begin{figure}[ht!]
%\centering
\includegraphics[width=0.5\textwidth,clip]{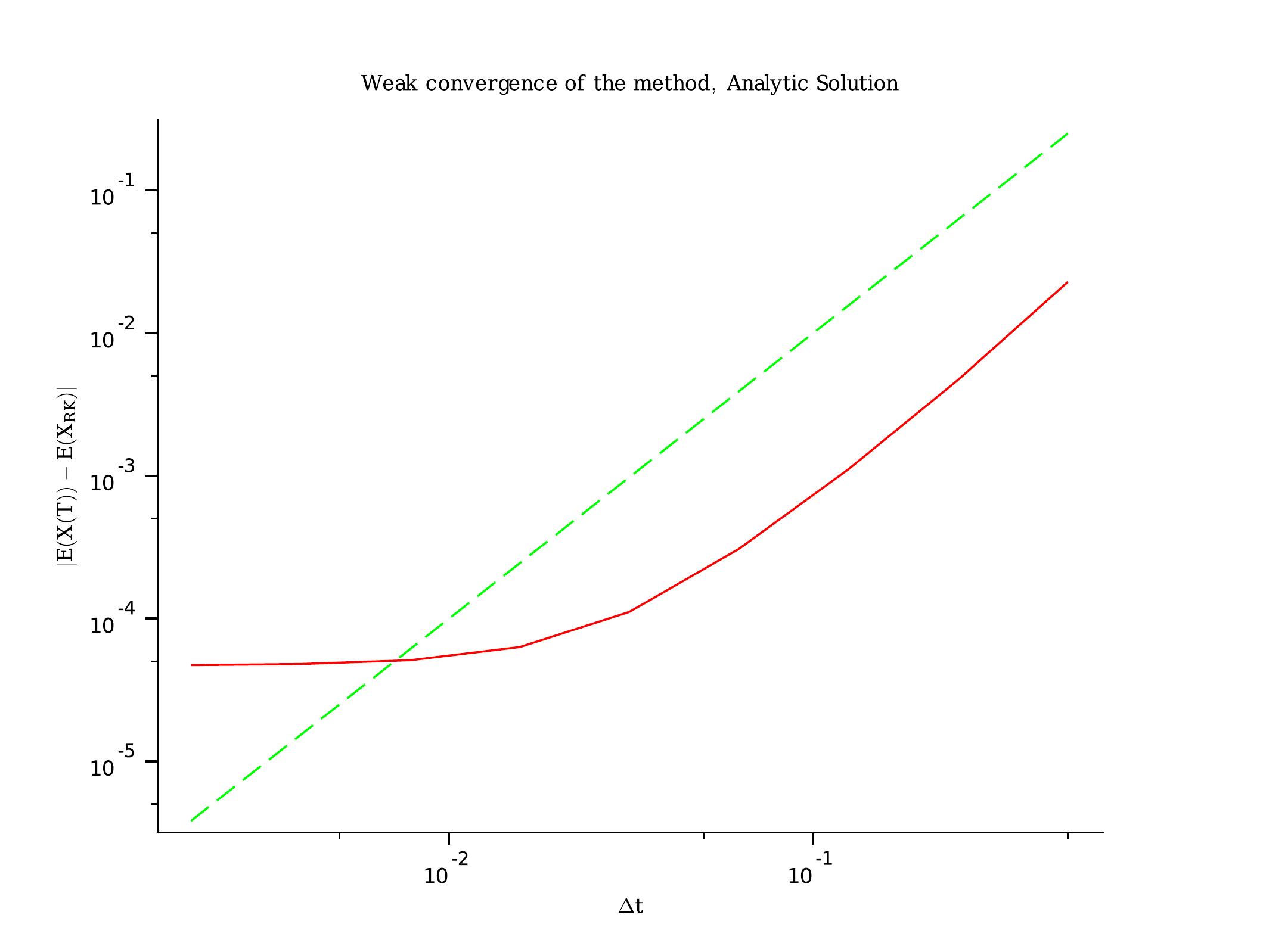}
\includegraphics[width=0.5\textwidth,clip]{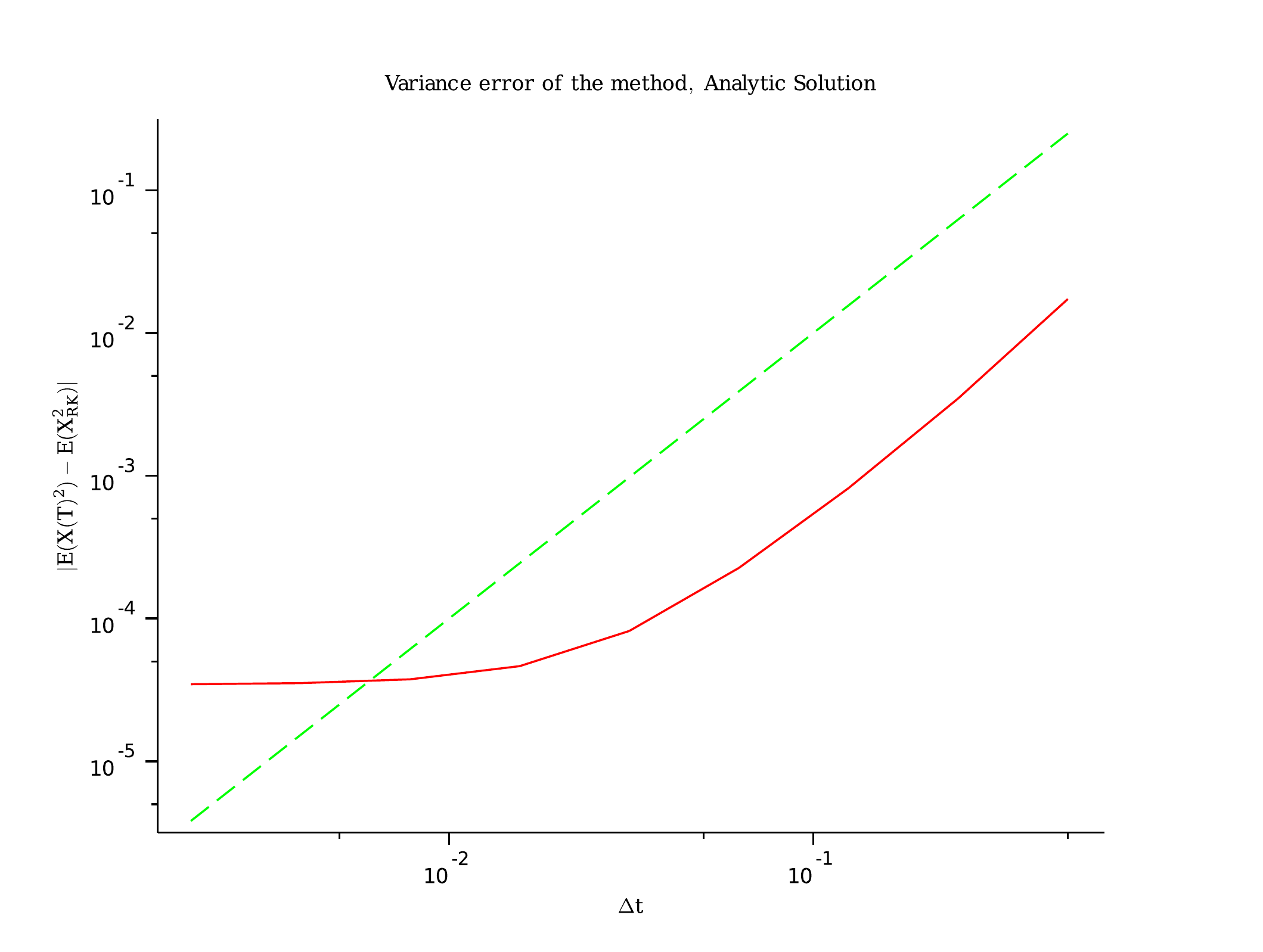}
\caption{$\sigma=0.001$-analytic case}
\label{analytic-test-0001}
\end{figure}

\begin{figure}[ht!]
%\centering
\includegraphics[width=0.5\textwidth,clip]{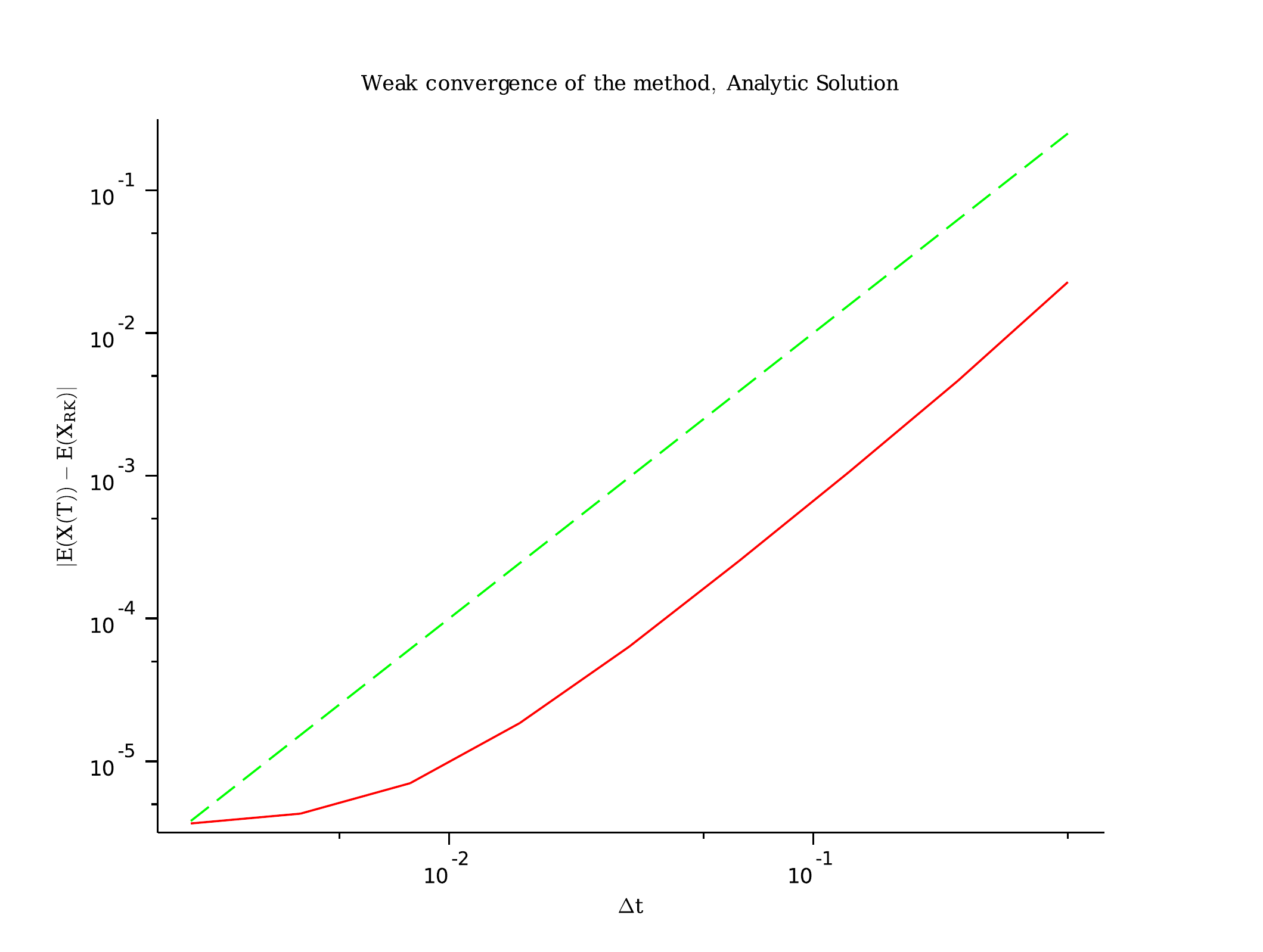}
\includegraphics[width=0.5\textwidth,clip]{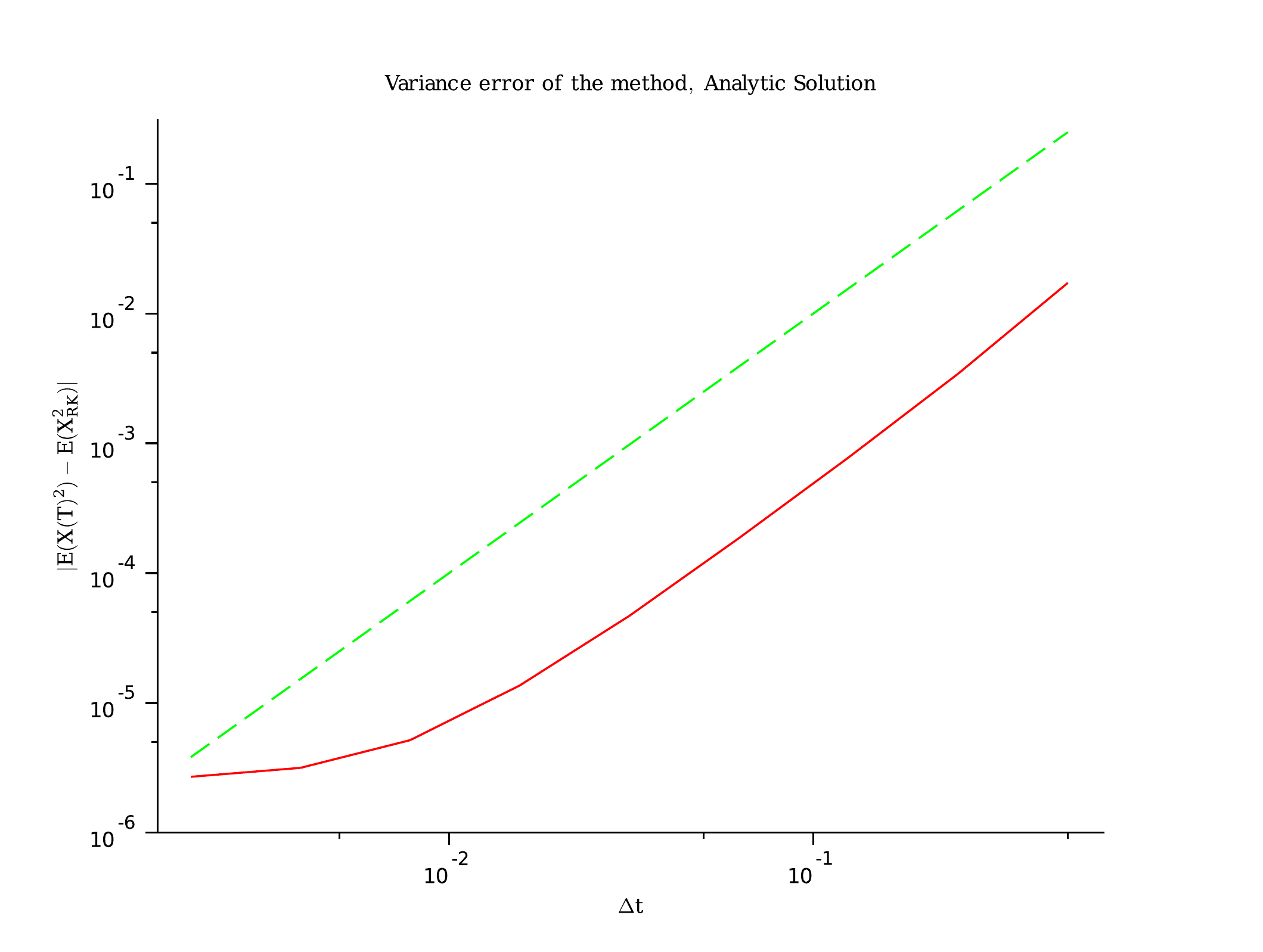}
\caption{$\sigma=0.0001$-analytic ase}
\label{analytic-test-00001}
\end{figure}

For the stochastic two body problem, we have no analytic solutions in order to test the convergence or not of the method. In order to study this problem, we follow a usual strategy in numerical analysis (see for example : we compute a {\it reference solution}, i.e. a solution compute with a very small time increment (in this case $2^{-10}$) and look for the difference between this reference solution and the algorithmic method. It gives a good idea of the order of the method. We provide a comparative test in the following for the Langevin equation. See Figures \ref{reference-test-0001} and \ref{reference-test-00001}.

\begin{figure}[ht!]
%\centering
\includegraphics[width=0.5\textwidth,clip]{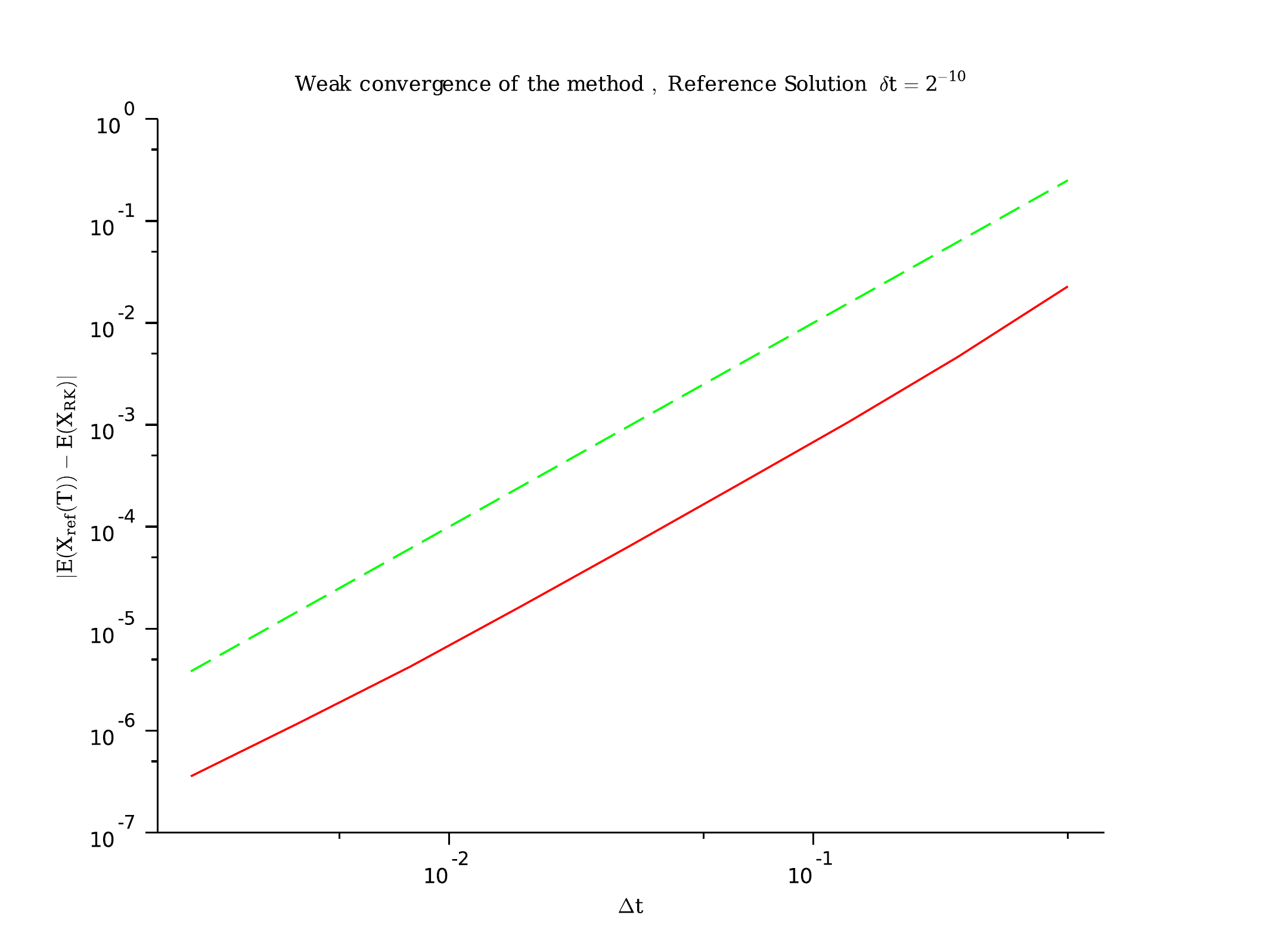}
\includegraphics[width=0.5\textwidth,clip]{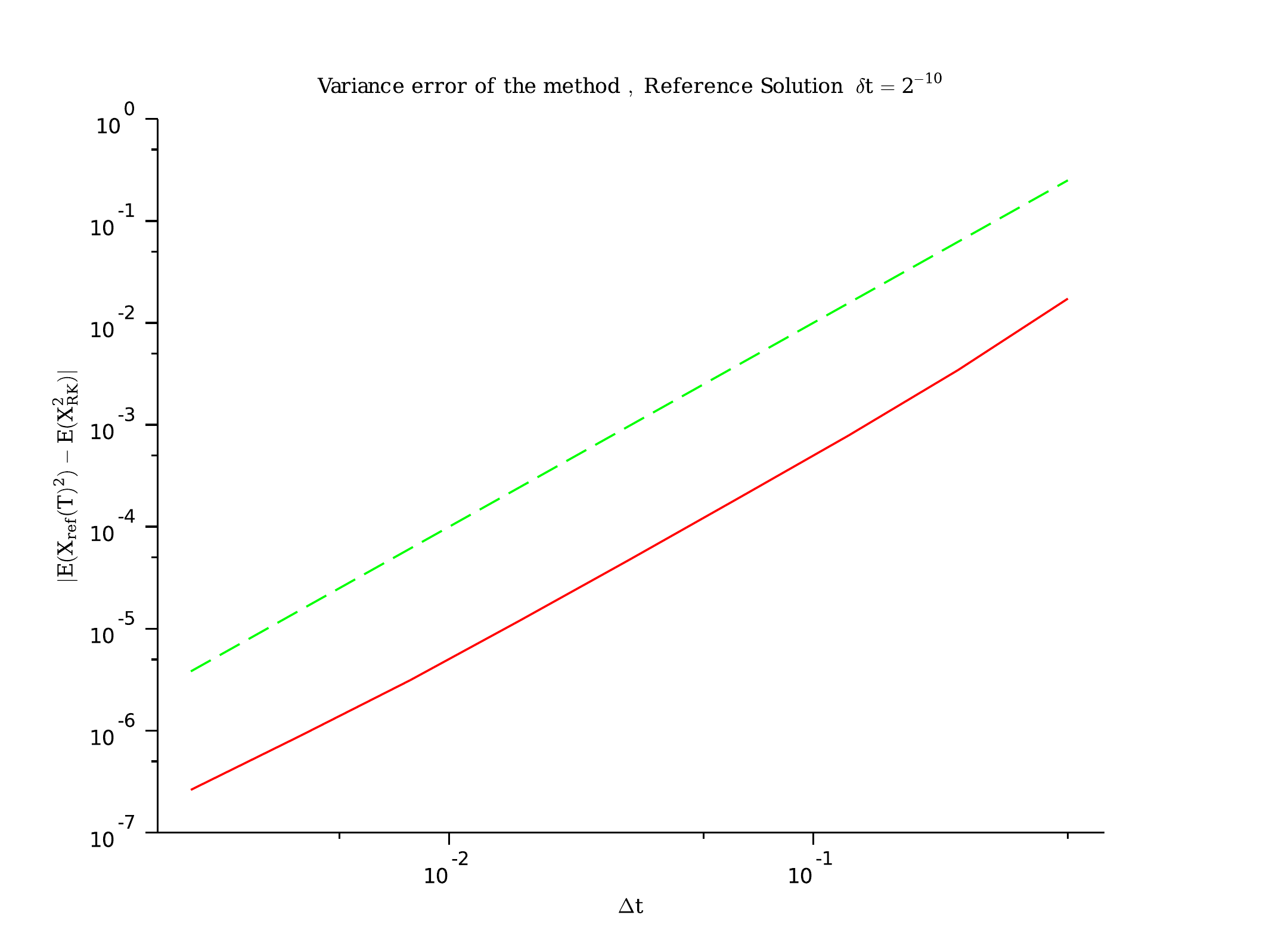}
\caption{$\sigma=0.001$-Reference solution case}
\label{reference-test-0001}
\end{figure}

\begin{figure}[ht!]
%\centering
\includegraphics[width=0.5\textwidth,clip]{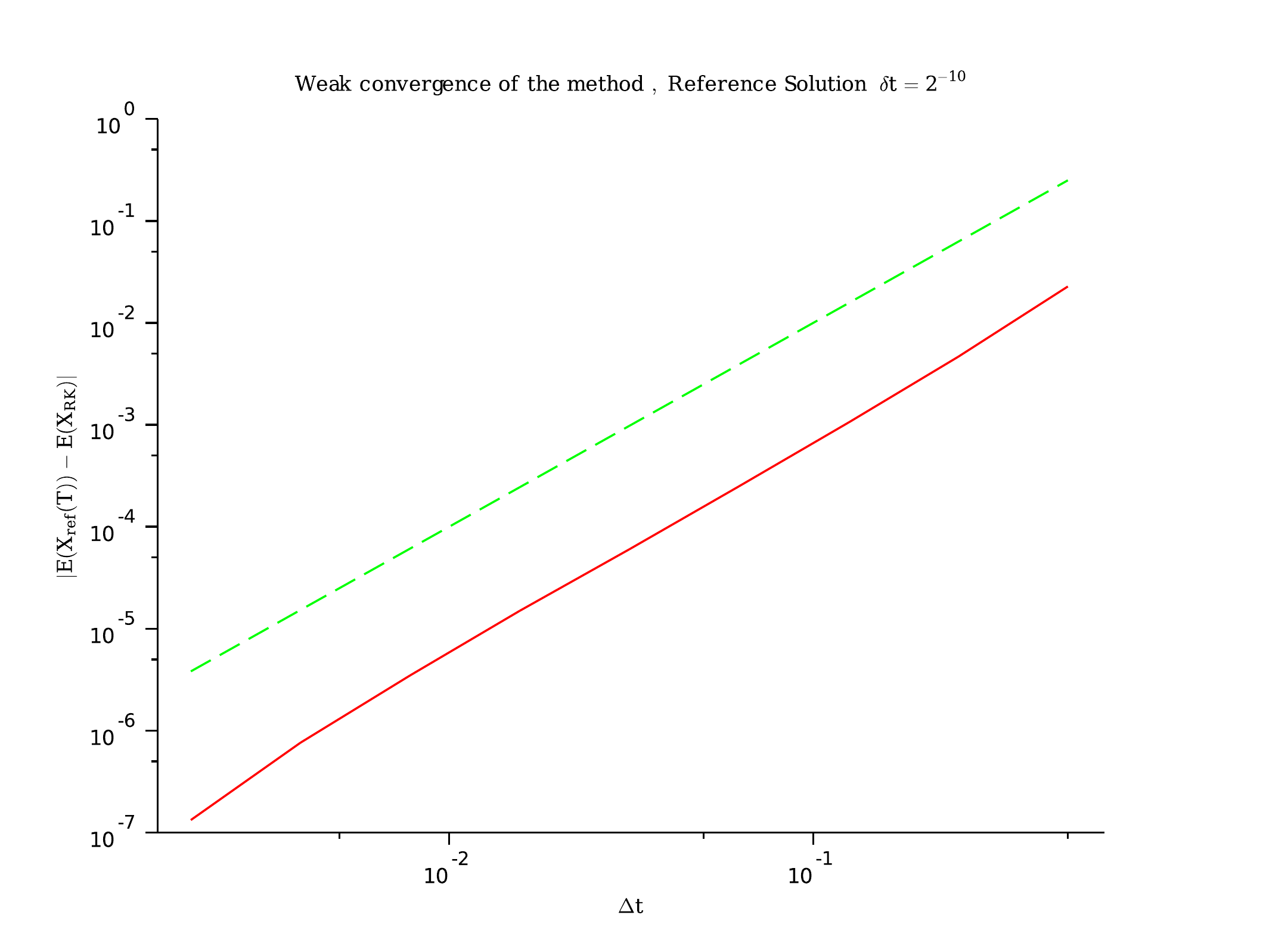}
\includegraphics[width=0.5\textwidth,clip]{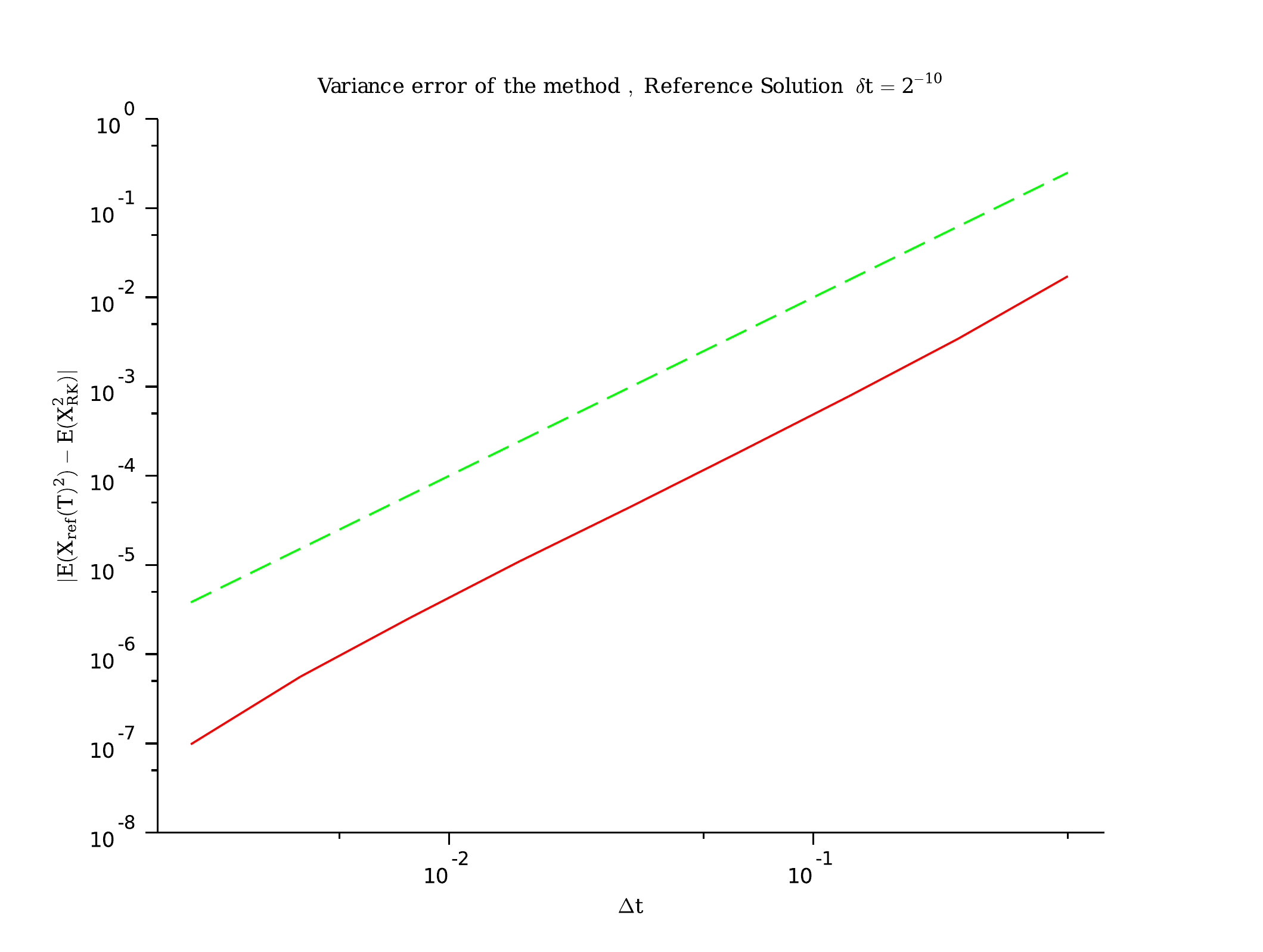}
\caption{$\sigma=0.0001$-Reference solution case}
\label{reference-test-00001}
\end{figure}

\subsection{Simulations for the stochastic two-body problem}

\subsubsection{Set of initial conditions}

Our simulations are made with the same initial conditions and integration time used by \cite{sharma}, which are :
\begin{eqnarray}
	r(0)&=& 1 \ AU ,\\
	\phi(0)&=& 1\ rad ,\\
	v(0)&=&0.01\ AU/TU ,\\
	\omega(0) &=& 1.1 \ rad/TU ,\\
	\sigma_{r}&=& 0.0121 \ TU^{-3/2} ,\\
	\sigma_{\phi}&=& 2.2 \times 10^{-4} \ AU.TU^{-3/2} ,
\end{eqnarray}
where AU is the Astronomical Unit which is the Earth-Sun distance and TU is the Time Unit which is approximately 58 days. These units are called canonical units (see \cite{bate71}).\\

The initials conditions are chosen such that the unperturbed motion is an ellipse and the diffusion constants $\sigma_{r}$ and $\sigma_{\phi}$ are chosen such that the stochastic perturbing force is proportional to $1/10$ of the gravitational force at the initial time. Numerical integration are performed over $\mathrm{15 TU}$ like in \cite{sharma}. 

\subsubsection{Numerical results}

The unperturbed trajectory as well as the perturbed one are plotted in Fig.~\ref{pierret:figsim} with color green and red respectively and we still use the same colors on figures to refer to the unperturbed and perturbed case. 
\begin{figure}[ht!]
\centering
\includegraphics[width=0.45\textwidth,clip]{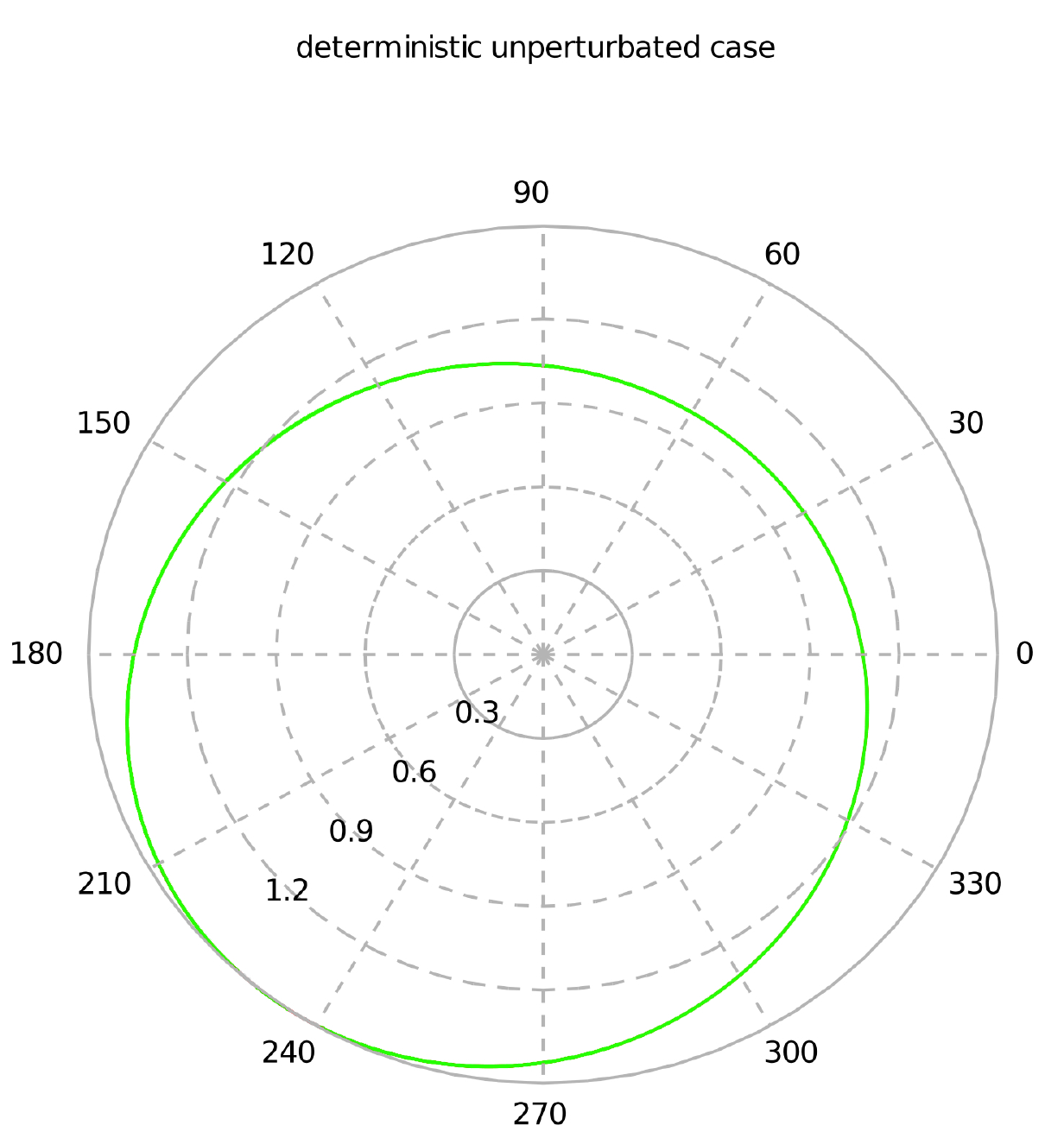}      
\includegraphics[width=0.45\textwidth,clip]{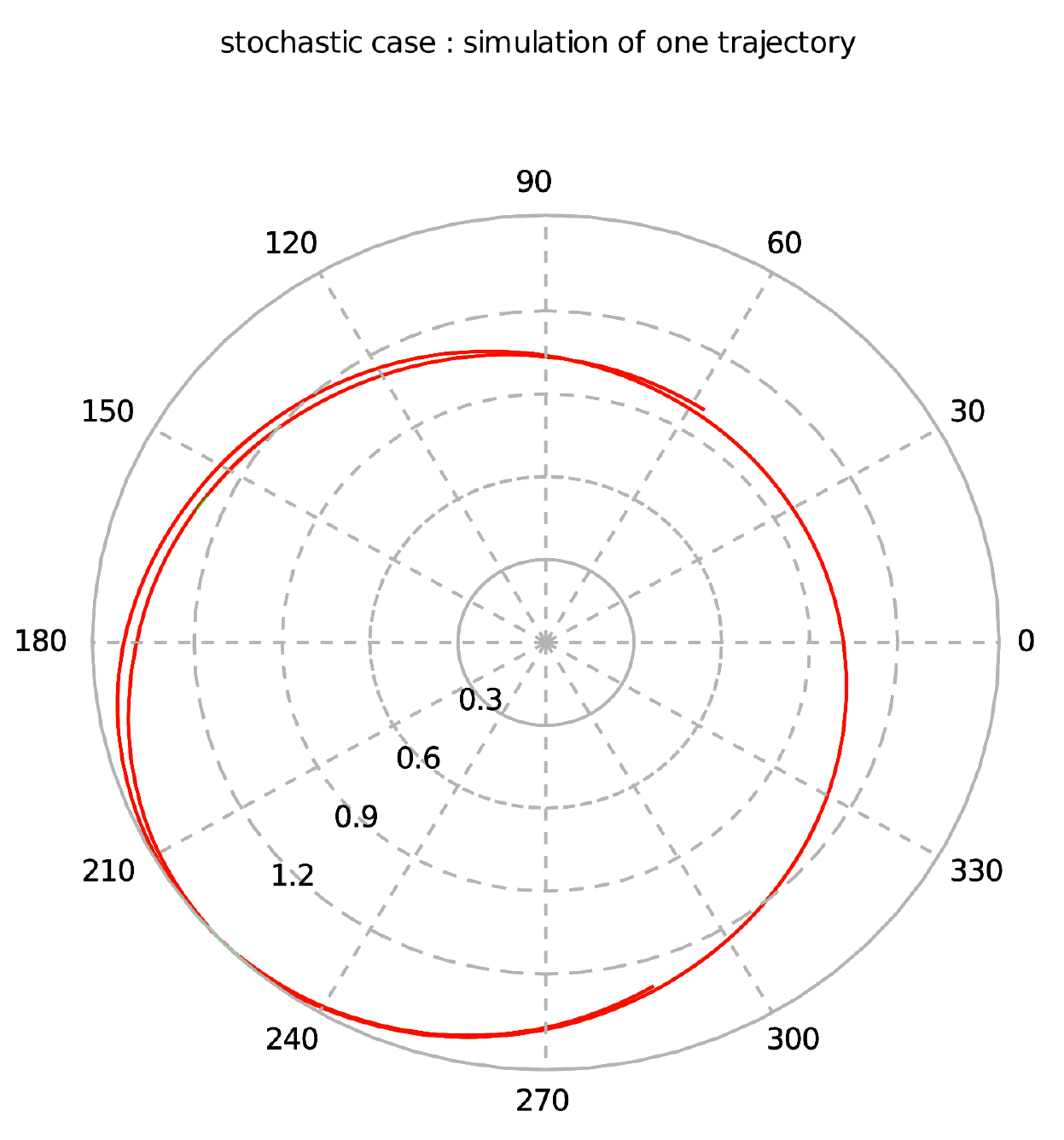}     %\label{key} 
%% Note the ABSENCE of the extension .pdf , .eps or .ps  !
\caption{{\bf Left:} Unperturbed case. {\bf Right:} Perturbed case. }
\label{pierret:figsim}
\end{figure}

Two other examples of solution are given in Figure \ref{fig-2b-other}.

\begin{figure}[ht!]
\centering
\includegraphics[width=0.45\textwidth,clip]{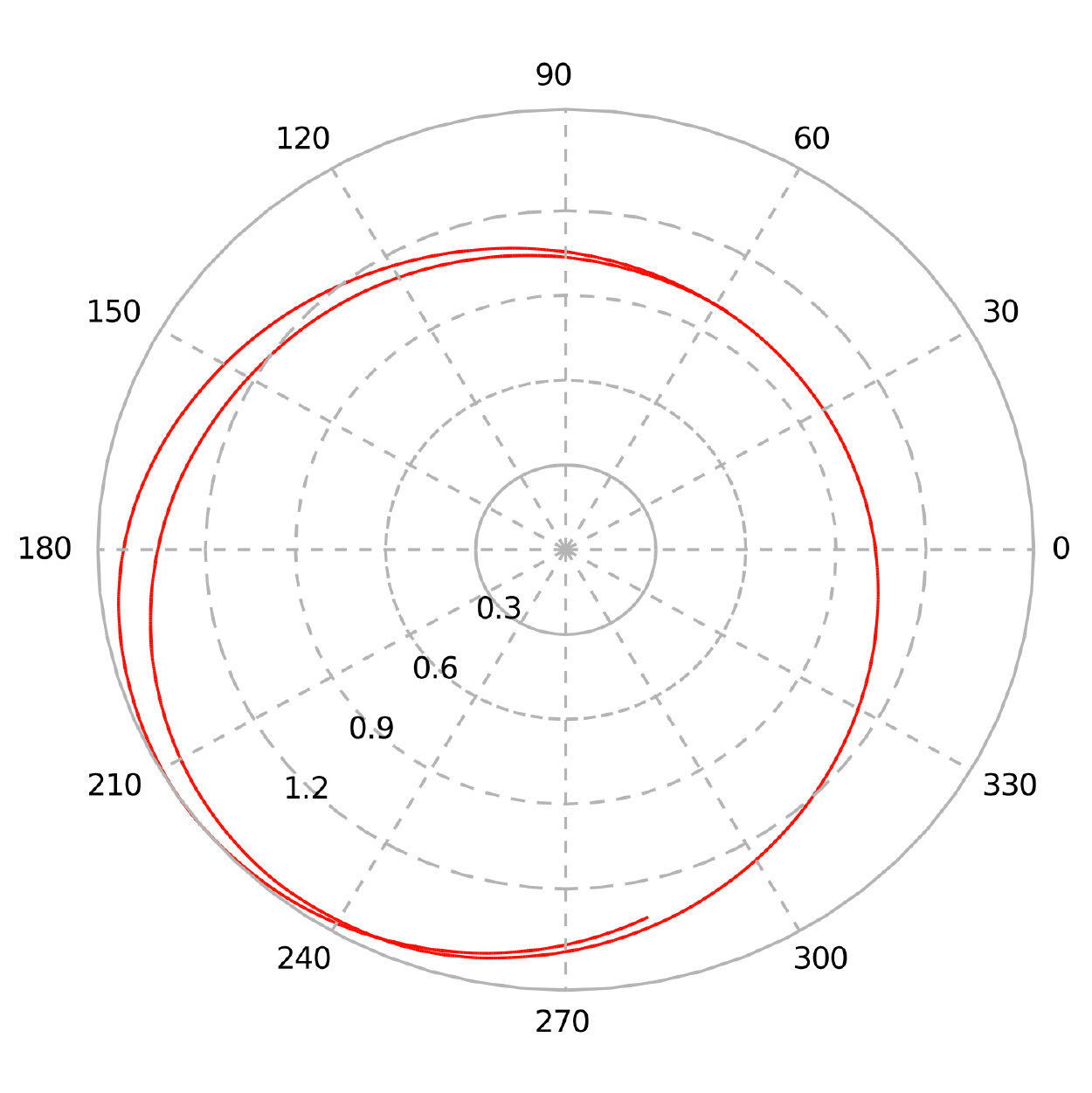}   
\includegraphics[width=0.45\textwidth,clip]{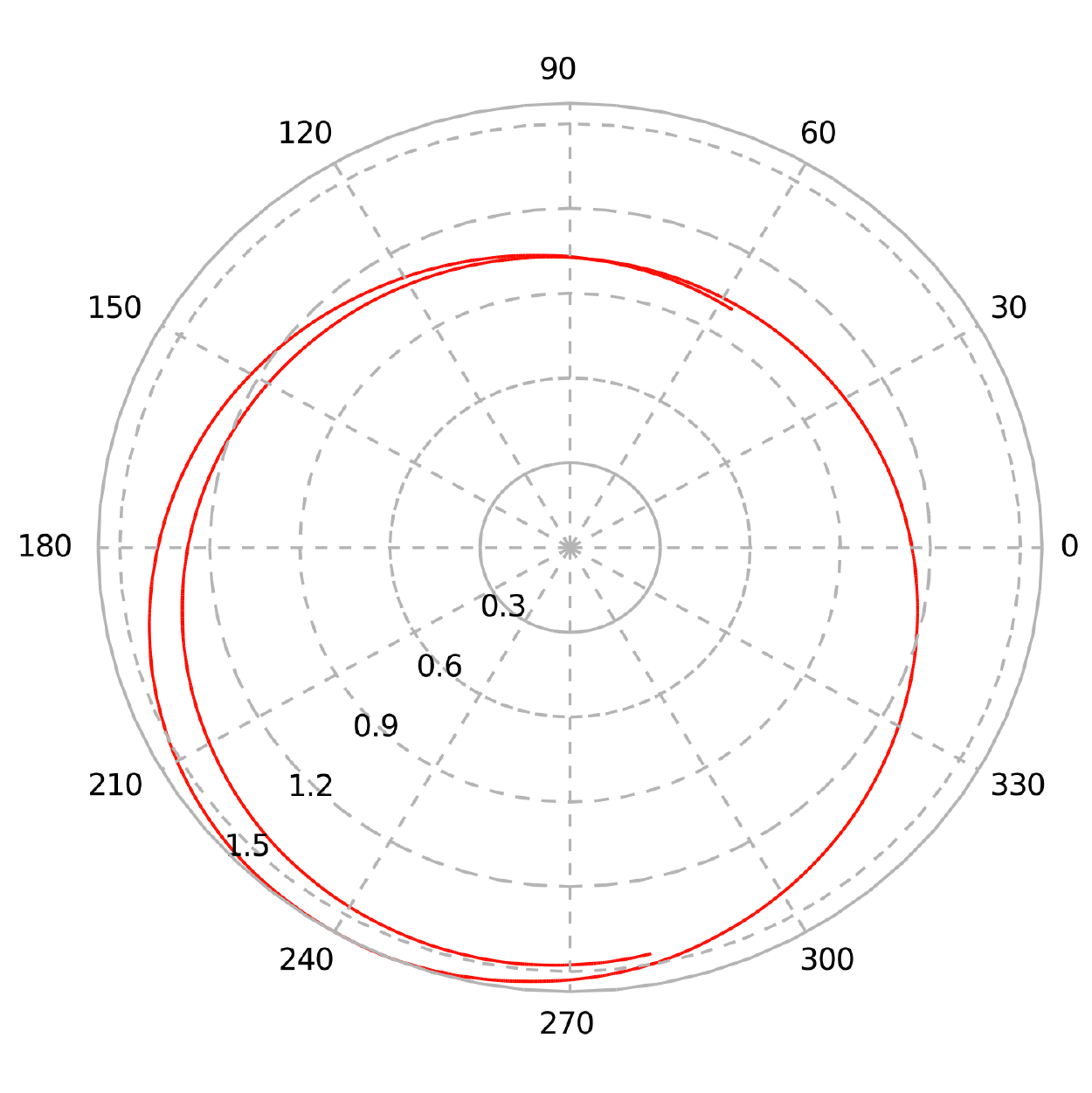}   
\caption{Other examples} 
\label{fig-2b-other}
\end{figure} 

The main feature of all these examples is a rapid divergence of the perturbed trajectory with respect to the unperturbed one despite the fact that the size of the perturbation is assumed very small. Such a fast divergence is not suspected in the context of classical perturbation theory and this idea justifies the fact to neglect many effects that are two small to induce a significant effect on the dynamics. However, as this example shows, if the nature of this effects does not enter in the framework of classical perturbation theory, as for example in the stochastic case, then one can not neglect the perturbation in the model as it induces significant different qualitative behaviours with respect to the unperturbed case.
  
\subsubsection{Accuracy and convergence of the numerical method}

As already pointed out in the Section concerning the stochastic Runge-Kutta methode of Kasdin and al. there exists up to now, no proof of the convergence of the algorithm. In the following, we give numerical evidences in the case of the stochastic two-body problem that this algorithm converges toward the exact solution of the equation in a weak sense. As we do not know an explicit form of the solution, we compute a {\it reference solution} with a very small time increment of order $2^{-10}$. The weak error is then computed with respect to this reference solution. See Figure \ref{weak_convergence_two_body-fig}.  

\begin{figure}[ht!]
\centering
\includegraphics[width=0.4\textwidth,clip]{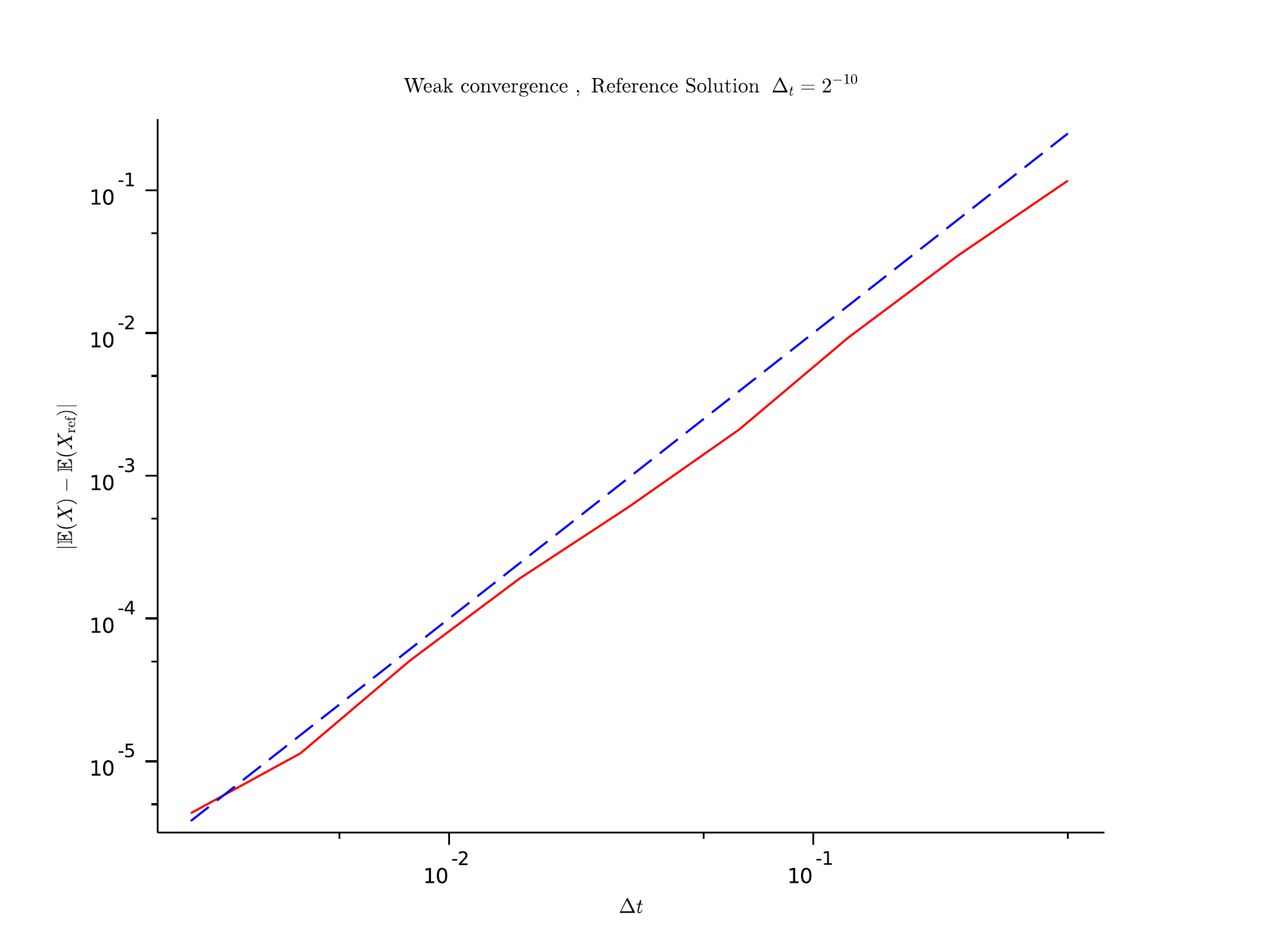}      
\caption{Numerical computation of the weak error}
\label{weak_convergence_two_body-fig}
\end{figure}

This simulation shows also that the stochastic Runge-Kutta is of order 2.\\ 

Moreover, the accuracy of the integrator can be tested by looking for the preservation of the weak first integral given by the angular momentum. Expectations are computed using a Monte Carlo method. Our result indicates a very good behavior of the integrator with respect to weak first integrals (see Fig.~\ref{pierret:figfirstint}).

\begin{figure}[ht!]
 \centering
 \includegraphics[width=0.4\textwidth,clip]{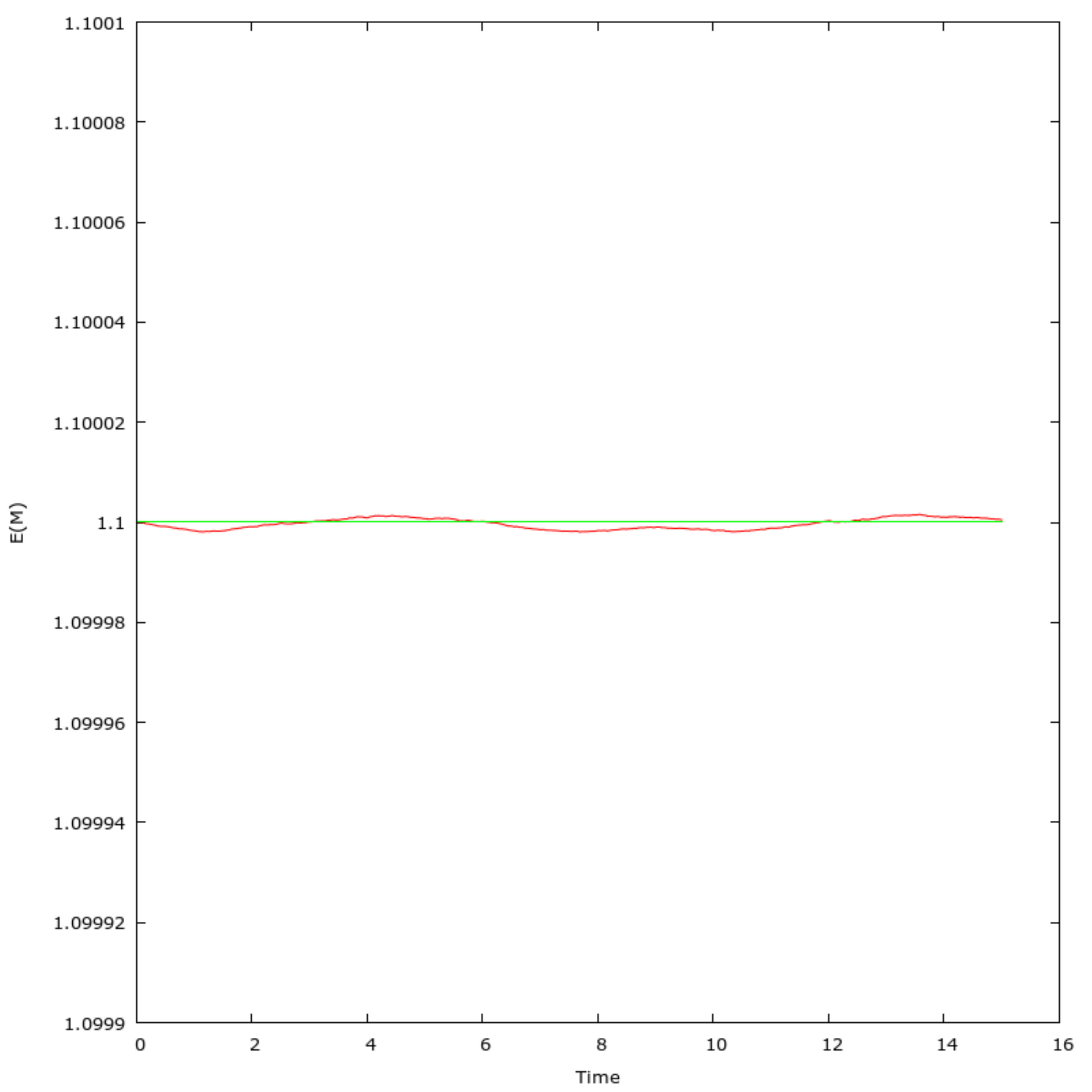}      
 \includegraphics[width=0.4\textwidth,clip]{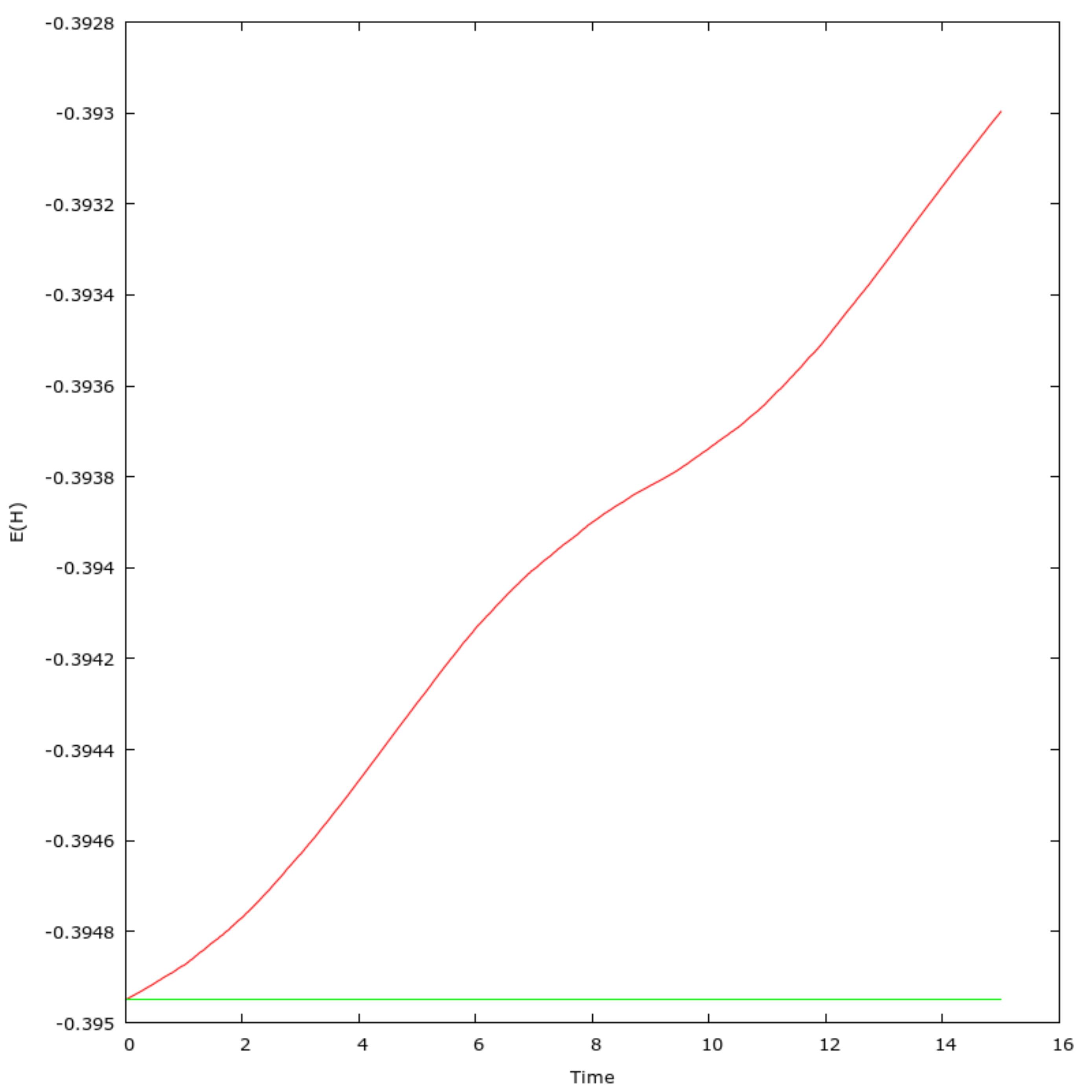}      
%% Note the ABSENCE of the extension .pdf , .eps or .ps  !
\caption{{\bf Left:} $\mbox{\rm E}(M(X_t))$. {\bf Right:} $\mbox{\rm E}(H(X_t))$. }
\label{pierret:figfirstint}
\end{figure}

\section{Stochastic planar Gauss equations}
\label{stochastic-gauss}
To study the variations of orbital elements we derive a stochastic version of the classical Gauss equations (see \cite{goldstein},p.96-103). This is done for a general version of our stochastic two-body problem allowing more general stochastic perturbation forces. Numerical results are then provided.

\subsection{Notations and the stochastic model}

We use notations of Section \ref{two-body-intro}. We denote by $m=\frac{M_S M_P}{M_S+M_P}$, $k=GM_S M_P$ and $\mu=k/m$.\\

We recall that \textbf{r} is the vector position from $M_S$ to $M_P$, \textbf{v} the velocity vector. The perturbed equations of motion are 
\begin{align}
d\textbf{r} &= \textbf{v}dt, \label{eqdvecr} \\
d\textbf{v} &= \left( -\frac{\mu}{r^3}\textbf{r} + \textbf{a}_P \right)dt, \label{eqdvecv}
\end{align}
where $\textbf{a}_P$ is the perturbing acceleration induced by the perturbing force $\textbf{F}$.

In $\{\textbf{e}_R,\textbf{e}_T,\textbf{e}_N\}$, the variation of the position vector $\textbf{r} = r \textbf{e}_R$ is given by
\begin{align}
d\textbf{r} &= dr \textbf{e}_R + r d\theta \textbf{e}_T .
\end{align}
We denote by $v$ the radial velocity and $w$ the transverse velocity defined by
\begin{align}
dr=vdt ,\\
d\theta=wdt .
\end{align}
The variation of the position vector is finally given by
\begin{align}
d\textbf{r} &= \left(v \textbf{e}_R + r w \textbf{e}_T \right)dt .
\end{align}

We identify the velocity vector $\textbf{v}$ as
\begin{align}
\textbf{v}=v \textbf{e}_R + r w \textbf{e}_T .
\end{align}
It follows that the variation of the velocity vector is given by
\begin{align}
d\textbf{v} = (dv-rw^2dt)\textbf{e}_R + (2vwdt+rdw)\textbf{e}_T .
\end{align}

In order to get the expression of the radial and transverse acceleration, we will precise the expression of the accelerating force $\textbf{a}_P$. \\

We assume that the perturbing acceleration is of stochastic nature modelled by a stochastic processes. Precisely, we make the following assumption :\\

{\bf Stochastic forces (S)} : {\it We assume that the velocity vector $d\textbf{v}_{P}$ satisfies the following Itô stochastic differential equation 
\begin{align}
d\textbf{v}_{P} \equiv \textbf{a}_P dt = \left(\begin{matrix} \bar{R} \\ \bar{T} \end{matrix}\right)dt + \left(\begin{matrix}
\tR_1 & \tR_2 \\ 
\tT_1 & \tT_2 
\end{matrix}\right)  \cdot dB_t
\end{align}

where $\bar{R}$ and $\bar{T}$ are in our problem, the deterministic part of the perturbation and $\tilde{R}=(\tR_1, \tR_2 )$ and $\tilde{T}=(\tT_1,\tT_2 )$ the purely stochastic part of the perturbation and $B_t=(B^R_t ,B^T_t)$ is a two-dimensional Brownian motion.}\\

Under assumption (S) the expression of the radial and transverse accelerations can be written as
\begin{align}
dv &= \left( rw^2-\frac{\mu}{r^2} + \bar{R} \right)dt + \tilde{R}\cdot dB_t ,\label{dv}\\
dw &= \left(-\frac{2vw}{r} + \frac{\bar{T}}{r} \right)dt + \frac{\tilde{T}}{r}\cdot dB_t \label{dw}.
\end{align}

The previous stochastic model is denoted by ($\star$) in the following. The main problem is to determine explicitly the set of stochastic differential equations governing the behaviour of the orbital elements $a,e$ and $\omega$. This is done in the next Section.

\subsection{Stochastic Gauss formula}

Using the notations of the previous Section, we obtain the following set of stochastic differential equations controlling the stochastic behaviour of the semi-major axis, eccentricity and pericenter for our stochastic model ($\star$). All the proofs are given in Appendix.

\begin{lemma}[Variation of semi-major axis] 
\label{lemme-major-axis}
The variation of the semi-major axis for the stochastic model ($\star$) is given by 
\begin{align}
		da&=\bigg[ \frac{2 a^{3/2}}{\sqrt{\mu (1-e^2)}}\left(e \sin f \bar{R} + (1+e\cos f) \bar{T}  \right) \\
		&+ \frac{a^2}{\mu} \left( \left(1+\frac{4e^2 \sin^2 f}{1-e^2}\right)\tilde{R}^2 + \left(1+\frac{4(1+e\cos f)^2}{1-e^2}\right)\tilde{T}^2 \right) \nonumber \\
		&+\frac{8a^2}{\mu (1-e^2)}e\sin f(1+e\cos f)\tilde{R}\cdot\tilde{T} \bigg]dt \nonumber \\
		&+ \frac{2 a^{3/2}}{\sqrt{\mu (1-e^2)}}\left(e \sin f \tilde{R} + (1+e\cos f) \tilde{T}  \right)\cdot dB_t . \nonumber \label{da}
\end{align}
\end{lemma}

The proof is given in Section \ref{proof-major-axis}.

\begin{lemma}[Variation of eccentricity] 
\label{lemme-eccentricity}
The variation of the eccentricity for the stochastic model ($\star$) is given by 
\begin{align}
de 	&= \bigg[ \sqrt{\frac{a(1-e^2)}{\mu}}\left(\sin f \bar{R} + (\cos f + \frac{e+\cos f}{1+e\cos f})\bar{T} \right) +  \frac{a(1-e^2)\cos^2 f}{2e\mu}\tR^2  \\
	&+\frac{a(1-e^2)}{\mu e}\left(2-\frac{\cos f}{2}\left(\frac{2+e\cos f}{1+e \cos f}\right) \left(\cos f + \frac{e+\cos f}{1+e\cos f}\right)\right)\tT^2 \nonumber \\
	&+ \frac{a(1-e^2)}{\mu e (1+e\cos f)}(e\sin^3 f - \sin 2f) \tR \cdot \tT\bigg]dt \nonumber \\
	&+\sqrt{\frac{a(1-e^2)}{\mu}}\left(\sin f \tR + (\cos f + \frac{e+\cos f}{1+e\cos f})\tT \right)\cdot dB_t \cdot \nonumber \label{de}
\end{align}
\end{lemma}

The proof is given in Section \ref{proof-eccentricity}.

\begin{lemma}[Variation of the pericenter] 
\label{lemme-pericenter}
The variation of the pericenter for the stochastic model ($\star$) is given by 
\begin{align}
	d\omega &=\bigg[ \sqrt{\frac{a(1-e^2)}{\mu}}\left( -\frac{\cos f}{e} \bR + \frac{\sin f}{e} \left(\frac{2+e\cos f}{1+e\cos f} \right)\bT\right)  \\
	&+ \frac{a(1-e^2)}{\mu e^2}\bigg( \frac{\sin 2f}{2} \tR^2 - \left(e+\cos f(2+e\cos f)^2 \right)\frac{\sin f}{(1+e\cos f)^2} \tT^2 + \left(\frac{2+e\cos f}{1+e\cos f}\right)\cos 2f \tR\cdot\tT \bigg) \bigg] dt  \nonumber \\ 
	&+\sqrt{\frac{a(1-e^2)}{\mu}}\left( -\frac{\cos f}{e} \tR + \frac{\sin f}{e} \left(\frac{2+e\cos f}{1+e\cos f} \right)\tT \right) \cdot dB_t .\nonumber
\end{align}
\end{lemma}

The proof is given in Section \ref{proof-pericenter}.

\subsection{Numerical results}

In this Section, we give numerical simulations comparing the behaviour of our stochastic Gauss equations and the numerical computation directly obtained from the trajectories of the stochastic two-body problem. In each case, one can see the very good agreement between the values obtained using our analytical formulas and the direct evaluations on a given solution. See Figures \ref{pierret:figa},\ref{pierret:fige} and \ref{pierret:figomega}. The energy, for which we have also derived an explicit formula is also well predicted. See Figure \ref{simul-energy}.

\begin{figure}
\centering

\begin{subfigure}[b]{0.45\textwidth}
 \centering
 \includegraphics[width=\textwidth,clip]{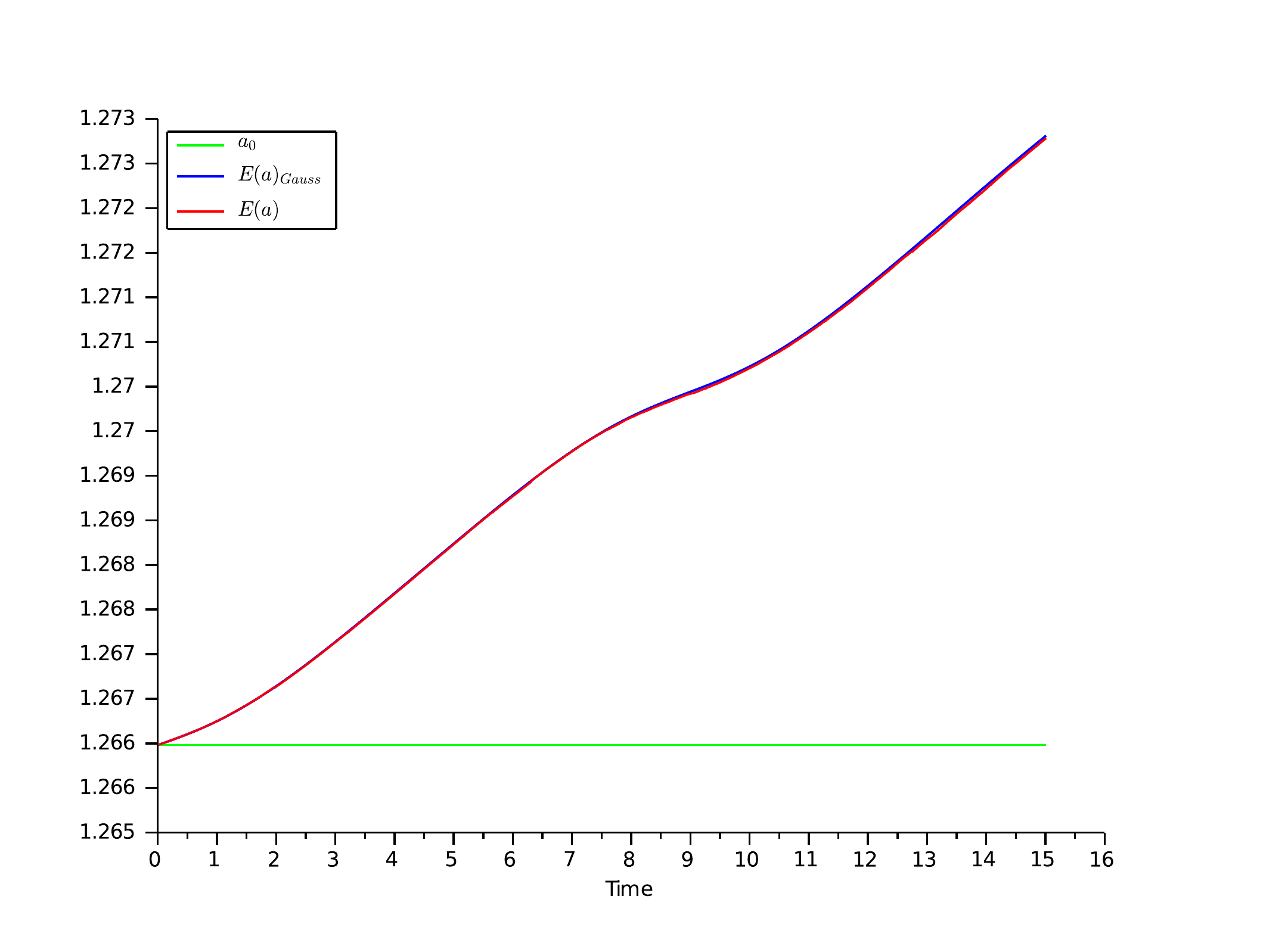}  
\caption{$a$ the semi-major axis}
\label{pierret:figa}
\end{subfigure}
~  
\begin{subfigure}[b]{0.45\textwidth}
 \centering
 \includegraphics[width=\textwidth,clip]{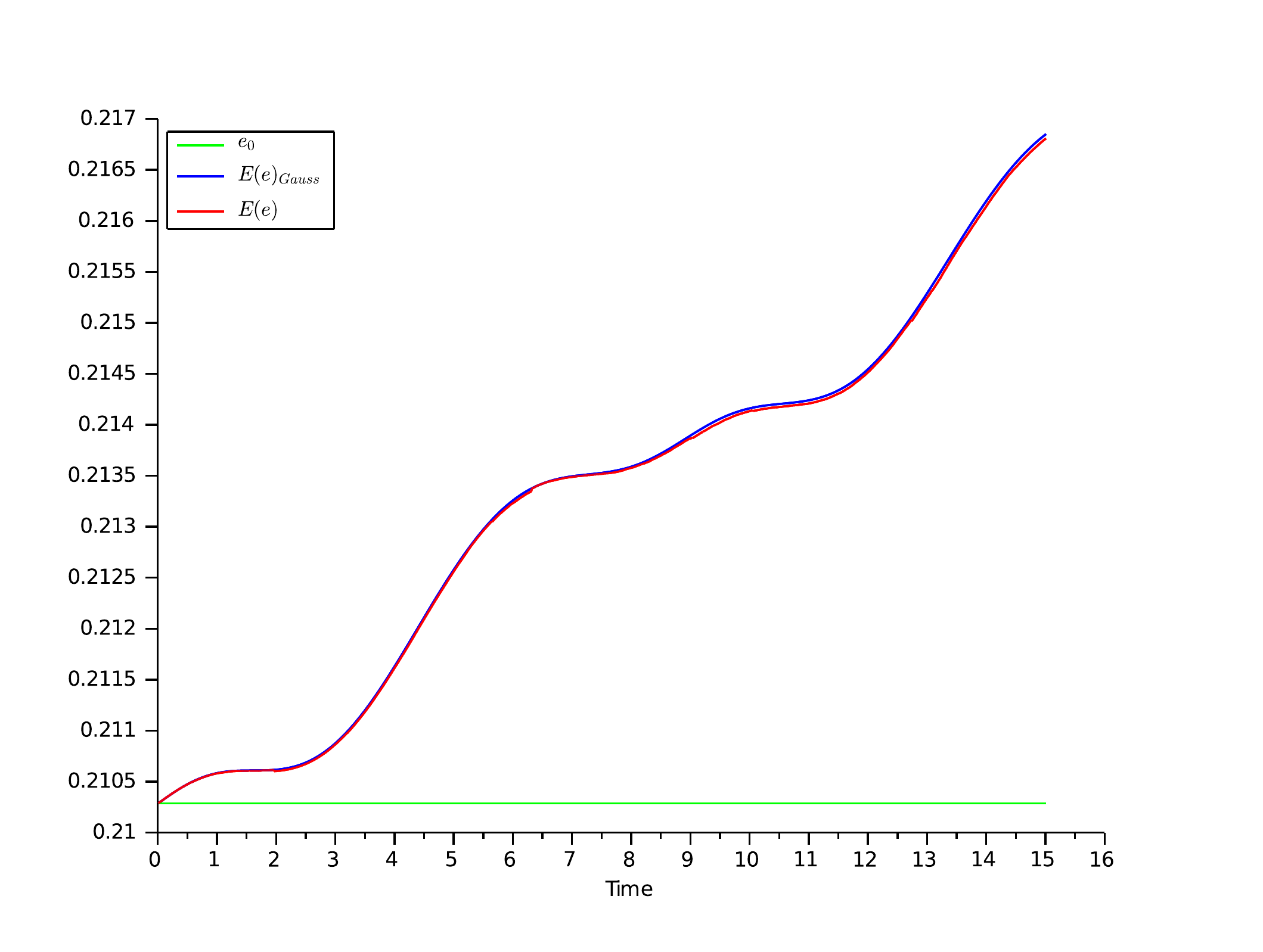}  
\caption{$e$ the eccentricity}
\label{pierret:fige}
\end{subfigure}    

\begin{subfigure}[b]{0.45\textwidth}
\centering
\includegraphics[width=\textwidth,clip]{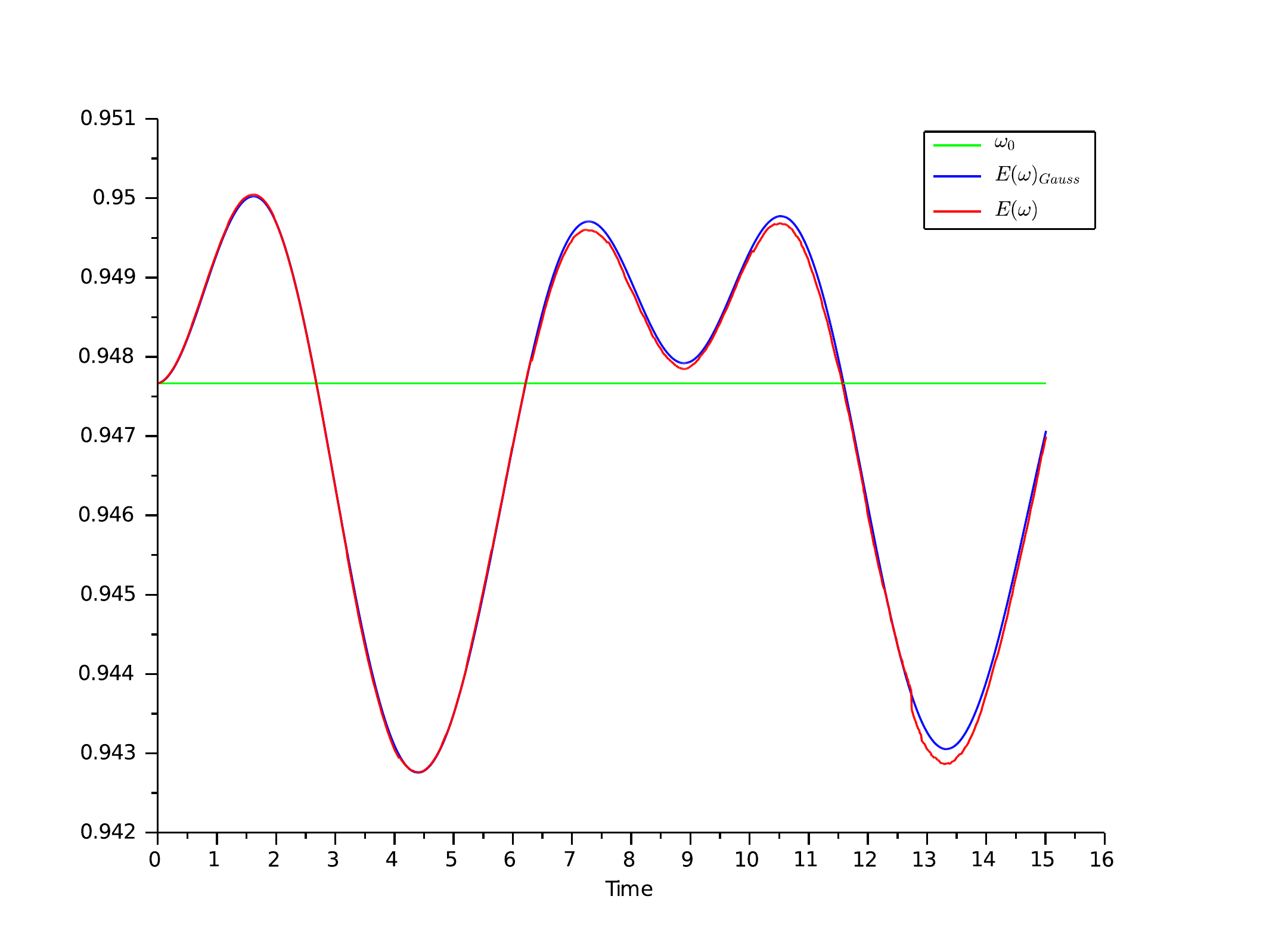}      
\caption{$\omega$ the pericenter angle. }
\label{pierret:figomega}
\end{subfigure}
~
\begin{subfigure}[b]{0.45\textwidth}
\centering
\includegraphics[width=\textwidth,clip]{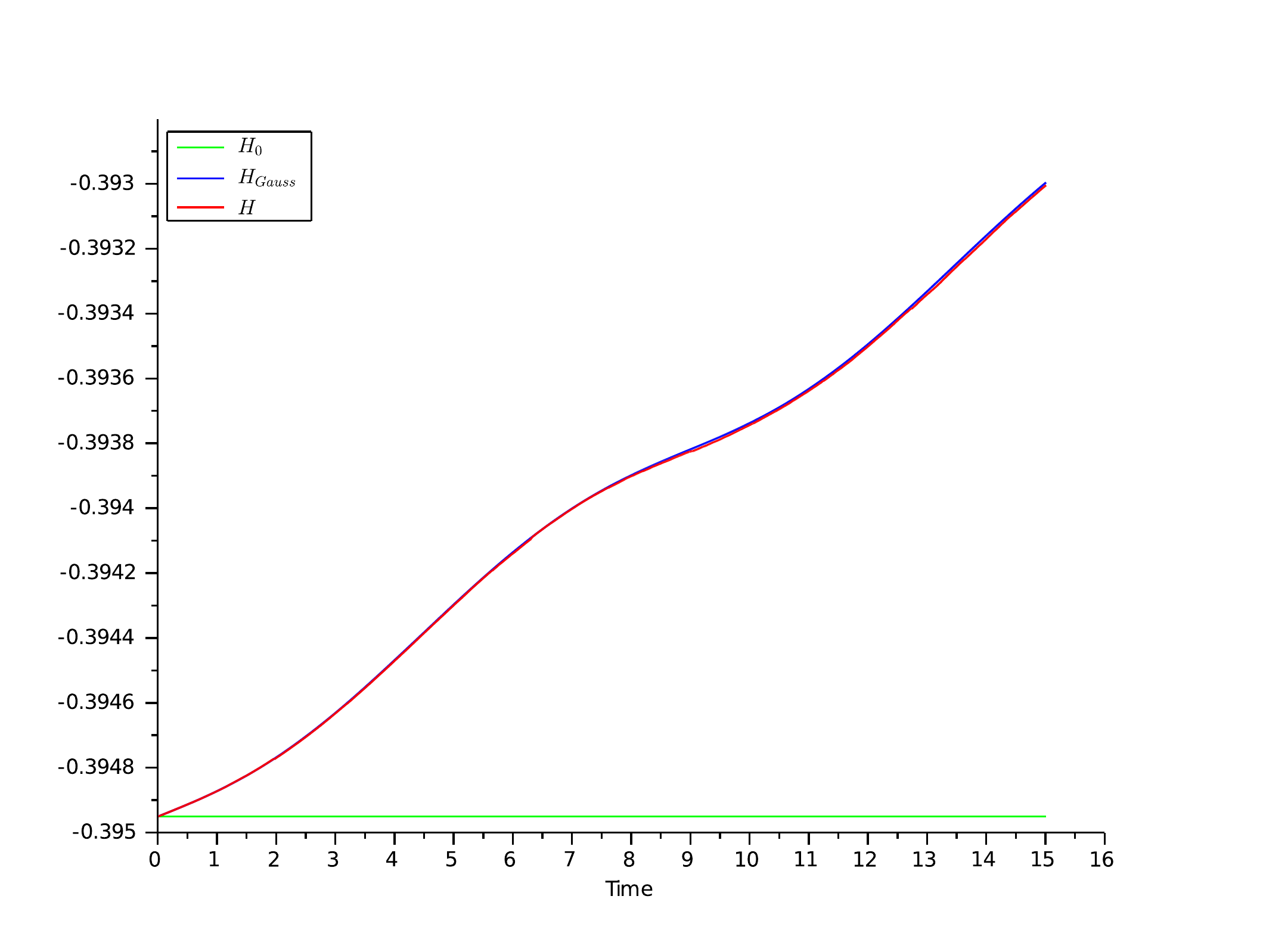}
\caption{$H$ the Energy}
\label{simul-energy}
\end{subfigure}

\end{figure}

\section{Conclusions}
\label{conclusion}
%%--------------------
The Sharma-Parthasarathy model displays a fast change of the dynamics with respect to the classical two-body problem despite the smallness of the stochastic perturbation. This result reinforces the necessity to take into account usually ignored stochastic phenomenon in order to obtain relevant predictions on the long term dynamical behaviour of dynamical systems.\\

As a consequence, the following list of open problems can be studied :
\begin{itemize}
\item Stochastic perturbations induced by the deformation of bodies. As a first step, we would like to study a $J_2$-problem (see \cite{j2}) with a random or stochastic $J_2$ constant and its influence on the rotation of the earth.

\item In order to perform simulations over a very long time, we need to construct high order stochastic Runge-Kutta type integrators.
\end{itemize}

% Optional acknowledgements
% -------------------------
%\begin{acknowledgements}
%
%\end{acknowledgements}

%%-----------------------------
%%   Bibliography
%%-----------------------------
%%

\begin{appendix}

\section{Preliminaries}

Let $B_t=(B_{t,1},\ \ldots,\ B_{t,m})^T$ denote $m$-dimensional Brownian motion. We can form the following $n$ It\^{o} processes

\begin{align}
\left\{\begin{array}{l}
dX_{1}=\bar{X}_{1}dt+\tilde{X}_{11}dB_{1}+\cdots+\tilde{X}_{1m}dB_{m}\\
\vdots  \\
dX_{n}=\bar{X}_{n}dt+\tilde{X}_{n1}dB_{1}+\cdots+\tilde{X}_{nm}dB_{m}
\end{array}\right. 
\end{align}

for $(1\leq i\leq n,\ 1\leq j\leq m)$. Or, in matrix notation simply

\begin{align}
dX(t)=\bar{X}dt+\tilde{X} \cdot dB_t , \label{nstocproc}
\end{align}

where

\begin{align*}
X(t)=\left(\begin{array}{c}
X_{1}(t)\\
\vdots \\
X_{n}(t)
\end{array}\right),
\ \bar{X}=\left(\begin{array}{c}
\bar{X}_{1}\\
\vdots \\
\bar{X}_{n}
\end{array}\right),
\ \tilde{X}=\left(\begin{array}{ccc}
\tilde{X}_{11} & \cdots & \tilde{X}_{1m}\\
\vdots  &  & \vdots \\
\tilde{X}_{n1} & \cdots & \tilde{X}_{nm}
\end{array}\right),
\ dB_t=\left(\begin{array}{c}
dB_{t,1}\\
\vdots \\
dB_{t,m}
\end{array}\right) .
\end{align*}

In what follow, we will always use the same notations as (\ref{nstocproc}) to describe a stochastic process and in order to simplify expressions we will omit the dependency.

\subsection{Energy and Angular momentum}

As it will be shown, we will need the variation of the angular momentum $M$ and the energy $H$ in order to compute the variation of the semi-major axis $a$, the eccentricity $e$ and the pericenter angle $\omega$.

We use the Itô formula to compute the variation of $M$ and $H$ with the perturbing radial acceleration (\ref{dv}) and the perturbing tangential acceleration (\ref{dw}) and we obtain 

\begin{align}
dM &= m r \bar{T} dt + m r \tilde{T}\cdot dB_t, \\
dH &= m \left( v\bar{R} +rw\bar{T}+\frac{\tilde{R}^2}{2}+\frac{\tilde{T}^2}{2}\right)dt + m \left(v\tilde{R} + rw \tilde{T}\right) \cdot dB_t.
\end{align}

In order to obtain the expression of $dM$ and $dH$ in term of orbital elements, we use the following formula which relate orbitals elements to $r,v$ and $w$ (see \cite{burns1976} Eq. (4)-(10)-(11)-(16))

\begin{align}
r&=\frac{a(1-e^2)}{1+e\cos f}, \label{rorb}\\
v&=\sqrt{\frac{\mu}{a(1-e^2)}}e\sin f, \label{vorb}\\
w&=\frac{\sqrt{\mu}}{a^{3/2}\sqrt{1-e^2}}(1+e\cos f)^2,\label{worb} \\
M&=\sqrt{mka(1-e^2)}\label{Morb}.
\end{align}

Finally using Eq. (\ref{rorb})-(\ref{vorb})-(\ref{Morb}) we get

\begin{align}
dM &= m\frac{a(1-e^2)}{1+e\cos f} \bar{T} dt + m\frac{a(1-e^2)}{1+e\cos f} \tilde{T}\cdot dB_t, \label{dH} \\
dH &= m\left[\sqrt{\frac{\mu}{a(1-e^2)}} \left(e\sin f \bar{R} +(1+e\cos f)\bar{T}\right)+\frac{\tilde{R}^2+\tilde{T}^2}{2}\right]dt \\ 
 	&+ m\sqrt{\frac{\mu}{a(1-e^2)}} \left( e\sin f \tilde{R} +(1+e\cos f)\tilde{T} \right) \cdot dB_t \nonumber. \label{dE}
\end{align}

In what follow we will always use Eq. (\ref{rorb})-(\ref{vorb})-(\ref{Morb}) to simplify terms.

\section{Proof of Lemma \ref{lemme-major-axis}}
\label{proof-major-axis}

The semi major axis is related to the energy by (see \cite{burns1976} Eq. (17))
\begin{align}
a=-\frac{k}{2H}.
\end{align}
Differentiating this equations yields to
\begin{align}
da = \frac{k}{2 H^2}dH - \frac{k}{2 H^3} \tilde{H}\cdot\tilde{H} dt,
\end{align}

Using the expression of the variation of the energy $H$ we get the expression of $da$.

\section{Proof of Lemma \ref{lemme-eccentricity}}
\label{proof-eccentricity}

The proof use the relation between the angular momentum $M$ and the energy $H$ given by (see \cite{burns1976} Eq. (18))

\begin{align}
	e=\sqrt{1+\frac{2M^2H}{mk^2}}
\end{align}

Differentiating this equation yields

\begin{align*}
de=\frac{2MH}{emk^2}dM+\frac{M^2}{emk^2}dH + \left(\frac{H}{e^3mk^2}\tilde{M}\cdot\tilde{M} - \frac{M^4}{2e^3m^2k^4}\tilde{H}\cdot\tilde{H} + \frac{2M(M^2H+mk^2)}{e^3m^2k^4}\tilde{M}\cdot\tilde{H}\right)dt.
\end{align*}

Firstly, notice that

\begin{align*}
1+e\cos f-\frac{(1-e^2)}{1+e\cos f}=e(\cos f + \frac{e+\cos f}{1+e\cos f})
\end{align*}

then

\begin{align*}
\frac{2MH}{emk^2}dM+\frac{M^2}{emk^2}dH = &\left[ \sqrt{\frac{a(1-e^2)}{\mu}}\left(\sin f \bar{R} + (\cos f + \frac{e+\cos f}{1+e\cos f})\bar{T} \right) + \frac{a(1-e^2)}{2e\mu} \left(\tR^2+\tT^2\right)\right]dt \\
& +\sqrt{\frac{a(1-e^2)}{\mu}}\left(\sin f \tR + (\cos f + \frac{e+\cos f}{1+e\cos f})\tT \right)\cdot dB_t .
\end{align*}

Secondly using the expression of $dM$ and $dH$ we have 

\begin{align*}
\tilde{M}\cdot\tilde{M} &= m^2\frac{a^2(1-e^2)^2}{(1+e\cos f)^2}\tT^2, \\
\tilde{H}\cdot\tilde{H} &= m^2\frac{\mu e^2\sin^2 f}{a(1-e^2)}\tR^2 + m^2\frac{\mu(1+e\cos f)^2}{a(1-e^2)}\tT^2+m^2\frac{2e\mu\sin f (1+e\cos f)}{a(1-e^2)}\tR\cdot\tT,\\
\tilde{M}\cdot\tilde{H} &= m^2\sqrt{\mu a(1-e^2)}\tT^2+m^2\frac{e\sin f \sqrt{\mu a(1-e^2)}}{1+e\cos f}\tR\cdot\tT.
\end{align*}

Finally after some simplifications we obtain the result.

\section{Proof of Lemma \ref{lemme-pericenter}}
\label{proof-pericenter}

The pericenter angle is defined by the relation (see \cite{goldstein} p.102-105)
\begin{align}
\tan \omega=\frac{A_y}{A_x}
\end{align}
or equivalently by
\begin{align}
\omega=\arctan \frac{A_y}{A_x} \label{omega}.
\end{align}
By definition of the Laplace-Runge-Lenz vector (Eq. \ref{runge-lenz}), the expression of his components are
$$
\left .
\begin{array}{lll}
A_x & = & m \left(\cos \phi \left(m r^3 w^2-k\right)+m r^2 w v\sin \phi\right), \\
A_y & = &m \left(\sin \phi \left(m r^3 w^2-k\right)-m r^2 w v\cos \phi\right).
\end{array}
\right .
$$
Using the Itô formula, the variations of $A_x$ and $A_y$ are given by
$$
\left .
\begin{array}{lll}
dA_x & = & \bigg[m^2 (r^2 w \sin \phi \bR + (2 r^2 w \cos \phi+rv \sin \phi)\bT ) + m^2 r \tilde{T}^2 \cos \phi +m^2r\sin \phi \tR\cdot\tT \bigg]dt \\
 & & + m^2 (r^2 w \sin \phi \tR + (2 r^2 w \cos \phi+rv \sin \phi)\tT )\cdot dB_t , \\
dA_y & = & \bigg[m^2 (-r^2 w \cos \phi\bR +(2 r^2 w \sin \phi-rv \cos \phi)\bT) +  m^2 r \tilde{T}^2 \sin \phi -m^2r\cos \phi \tR\cdot\tT \bigg]dt \nonumber \\
& & +m^2 (- r^2 w \cos \phi\tR +(2 r^2 w \sin \phi-rv \cos \phi)\tT)\cdot dB_t ,
\end{array}
\right .
$$
with
$$
\left .
\begin{array}{lll}
\tAx & = & m^2 (r^2 w \sin \phi \tR + (2 r^2 w \cos \phi+rv \sin \phi)\tT ) , \\
\tAy & = & m^2 (-r^2 w \cos \phi \tR + (2 r^2 w \sin \phi-rv \cos \phi)\tT ) .
\end{array}
\right .
$$
We now use the Itô formula on Eq.(\ref{omega}) which defines $\omega$ and we get
$$
d\omega = \frac{A_x dA_y-A_y dA_x}{A_x^2+A_y^2} + \left(\frac{\left(A_y^2-A_x^2\right) \tAx\cdot \tAy+A_xA_y\left(\tAx\cdot \tAx-\tAy\cdot \tAy\right)}{\left(A_x^2+A_y^2\right)^2}\right)dt.
$$
Firstly we detail some terms :
$$
\left .
\begin{array}{lll}
A_x^2+A_y^2  & = & e^2m^2k^2, \\
A_y^2-A_x^2 & = & -e^2m^2k^2\cos{2\omega}, \\
A_x A_y & = &  e^2m^2k^2\frac{\sin{2\omega}}{2}.
\end{array}
\right .
$$
So we have :
$$
\frac{A_x dA_y-A_y dA_x}{A_x^2+A_y^2} = \frac{mr\sin f \tT^2}{ek}dt+\frac{m}{ek}\left(r^2w\sin f \tR + (2r^2w\sin f -rv\cos f)\tT - r\cos f \tR\cdot\tT \right)\cdot dB_t .
$$
Secondly we have
$$
\left .
\begin{array}{lll}
\tAx \cdot \tAx & = & m^4 \left [ r^4 w^2 \sin^2 \phi \tR^2 + (2 r^2 w \cos \phi+rv \sin \phi)^2 \tT^2 \right . \\
& & \left . +2r^2w\sin \phi\left( 2r^2w\cos \phi +rv\sin \phi \right)\tR\cdot\tT \right ], \\
\tAy \cdot \tAy & = & m^4 \left [ r^4 w^2 \cos^2 \phi \tR^2 + (2 r^2 w \sin \phi-rv \cos \phi)^2\tT^2 \right . \\
& & \left . -2r^2w\cos \phi\left( 2r^2w\sin \phi -rv\cos \phi \right)\tR\cdot\tT \right ] , \\
\tAx \cdot \tAy & = & m^4 \bigg[ -\frac{r^4w^2\sin 2\phi}{2}\tR^2 + (2r^4w^2\cos\phi+rv\sin\phi)(2r^4w^2\sin\phi-rv\cos\phi)\tT^2 \\
 & & -r^2w\left(2r^2w\cos 2\phi +rv\sin 2\phi \right)\tR\cdot\tT\bigg], \\
\tAx^2-\tAy^2 & = & m^4\bigg[-r^4w^2\cos{2\phi}\tR^2 + \left( (2 r^2 w \cos \phi+rv \sin \phi)^2 - (2 r^2 w \sin \phi-rv \cos \phi)^2 \right)\tT^2 \\
 & & +2r^2w\left(2r^2w\sin 2\phi -rv\cos 2\phi \right)\tR\cdot\tT\bigg].
\end{array}
\right .
$$
We now detail the last term of $d\omega$,
$$
\left(A_y^2-A_x^2\right) \tAx\cdot \tAy+A_xA_y\left(\tAx\cdot \tAx-\tAy\cdot \tAy\right).
$$
In this term there are only factors of $\tR^2$, $\tT^2$ and $\tR\cdot\tT$. As a consequence, using the previous expressions the factor of $\tR^2$ is given by
$$
m^6e^2k^2r^4w^2\frac{\sin{2f}}{2},
$$
the factor of $\tT^2$ by
$$
m^6e^2k^2\left(2r^3wv\cos{2f}-(4r^4w^2-r^2v^2)\frac{\sin{2f}}{2}\right)
$$
and the factor of $\tR\cdot\tT$ by
$$
m^6e^2k^2\left(2r^4w^2\cos 2f + r^3wv\sin 2f \right).
$$
This leads to the following expression of $d\omega$ :
$$
\left .
\begin{array}{lll}
	d\omega & = & \bigg[ \sqrt{\frac{a(1-e^2)}{\mu}}\left( -\frac{\cos f}{e} \bR + \frac{\sin f}{e} \left(\frac{2+e\cos f}{1+e\cos f} \right)\bT\right)  \\
	& & + \frac{a(1-e^2)}{\mu e^2}\bigg( \frac{\sin 2f}{2} \tR^2 + \left(\frac{2e\sin f}{1+e\cos f}\cos 2f - \left(2-\frac{e^2\sin^2 f}{2(1+e\cos f)^2} \right) \sin 2f + \frac{e\sin f}{1+e\cos f}\right) \tT^2  \\
	&  & + \left(\frac{2+e\cos f}{1+e\cos f}\right)\cos 2f \tR\cdot\tT \bigg) \bigg] dt \nonumber \\
	& & +\sqrt{\frac{a(1-e^2)}{\mu}}\left( -\frac{\cos f}{e} \tR + \frac{\sin f}{e} \left(\frac{2+e\cos f}{1+e\cos f} \right)\tT \right) \cdot dB_t .\nonumber
\end{array}
\right .
$$
We can also simplify the term of $\tT^2$ as follows :
$$
\di\frac{2e\sin f}{1+e\cos f}\cos 2f - \left(2-\frac{e^2\sin^2 f}{2(1+e\cos f)^2} \right) \sin 2f = -\left(1+\frac{1}{1+e \cos f}\right) \left(\cos f+\frac{e+\cos f}{1+e \cos f}\right) \sin f 
$$
which gives
$$
-\left(1+\frac{1}{1+e \cos f}\right) \left(\cos f+\frac{e+\cos f}{1+e \cos f}\right) \sin f + \frac{e\sin f}{1+e\cos f} = -\frac{\left(e+\cos f (2+e \cos f)^2\right) \sin f}{(1+e \cos f)^2} .
$$
Finally, we obtain
$$
\left .
\begin{array}{lll}
	d\omega & = & \bigg[ \sqrt{\frac{a(1-e^2)}{\mu}}\left( -\frac{\cos f}{e} \bR + \frac{\sin f}{e} \left(\frac{2+e\cos f}{1+e\cos f} \right)\bT\right)  \\
	& & + \frac{a(1-e^2)}{\mu e^2}\bigg( \frac{\sin 2f}{2} \tR^2 - \left(e+\cos f(2+e\cos f)^2 \right)\frac{\sin f}{(1+e\cos f)^2} \tT^2 + \left(\frac{2+e\cos f}{1+e\cos f}\right)\cos 2f \tR\cdot\tT \bigg) \bigg] dt  \nonumber \\ 
	& & +\sqrt{\frac{a(1-e^2)}{\mu}}\left( -\frac{\cos f}{e} \tR + \frac{\sin f}{e} \left(\frac{2+e\cos f}{1+e\cos f} \right)\tT \right) \cdot dB_t ,
\end{array}
\right .
$$
which concludes the proof.

\end{appendix}

\end{document}